\documentclass[a4paper,UKenglish,cleveref, autoref, thm-restate]{lipics-v2021}

\pdfoutput=1 
\hideLIPIcs

\title{A Framework for Efficient Approximation Schemes on Geometric Packing Problems of $d$-dimensional Fat Objects}
\titlerunning{A Framework for Efficient Approximation Schemes on Geometric Packing Problems}

\author{Vítor Gomes Chagas}{Institute of Computing, University of Campinas, Brazil}{vitor.chagas@ic.unicamp.br}{https://orcid.org/0000-0002-6506-4174}{CNPq (Proc. 163645/2021-3), CAPES (Proc. 88882.329100/2014-01)}

\author{Elisa Dell'Arriva}{Institute of Computing, University of Campinas, Brazil}{elisa.arriva@ic.unicamp.br}{https://orcid.org/0000-0002-7505-5386}{CNPq (Proc. 161030/2021-1)}

\author{Flávio Keidi Miyazawa}{Institute of Computing, University of Campinas, Brazil}{fkm@ic.unicamp.br}{https://orcid.org/0000-0002-1067-6421}{CNPq (Proc.~313146/2022-5, and 404315/2023-2), FAPESP (Proc.~2015/11937-9, and 2022/05803-3)}

\authorrunning{V.\,G. Chagas, E. Dell'Arriva and F.\,K. Miyazawa}

\Copyright{Vítor G. Chagas, Elisa Dell'Arriva and Flávio K. Miyazawa} 



\relatedversion{}

\nolinenumbers

\usepackage[dvipsnames]{xcolor}
\usepackage[onelanguage, ruled, linesnumbered, noend, noline]{algorithm2e}
\makeatletter
\let\c@author\relax
\usepackage[dvipsnames]{xcolor}
\makeatother
\usepackage[
    backend=bibtex,
    bibstyle=ieee,
    citestyle=numeric-comp,
    sorting=nyt,
    natbib=true,
    mincitenames=1,
    maxcitenames=2,
    doi=false,
    eprint=false,
    isbn=false,
    url=false,]
{biblatex}
\usepackage{colortbl}
\usepackage{csquotes}
\usepackage{eqparbox}
\usepackage{float}
\usepackage[T1]{fontenc}
\usepackage{makecell}
\usepackage{multirow}
\usepackage{nicefrac}
\usepackage{readarray}
\usepackage{textcomp,gensymb}
\usepackage{thm-restate}
\usepackage{tikz}                
\usepackage{thm-restate}
\usepackage{physics}
\usepackage{scalerel}

\usepackage[theorem]{macros_general}
\usepackage{macros_specific}

\usetikzlibrary{
    angles,
    arrows, arrows.meta,
    babel,
    calc,
    math,
    positioning,
    quotes,
    shapes
}

\appto{\bibsetup}{\sloppy}
\DefineBibliographyStrings{english}{%
  andothers = {{et al}\adddot}
}

\crefname{algocfline}{line}{lines}

\allowdisplaybreaks{}

\begin{document}

\maketitle

\begin{abstract}
    We investigate approximation algorithms for several fundamental optimization problems on geometric packing.
    The geometric objects considered are very generic, namely $d$-dimensional convex fat objects.
    Our main contribution is a versatile framework for designing efficient approximation schemes for classic geometric packing problems.
    The framework effectively addresses problems such as the multiple knapsack, bin packing, multiple strip packing, and multiple minimum container problems.
    Furthermore, the framework handles additional problem features, including item multiplicity, item rotation, and additional constraints on the items commonly encountered in packing contexts.
    The core of our framework lies in formulating the problems as integer programs with a nearly decomposable structure. This approach enables us to obtain well-behaved fractional solutions, which can then be efficiently rounded. By modeling the problems in this manner, our framework offers significant flexibility, allowing it to address a wide range of problems and incorporate additional features.
    To the best of our knowledge, prior to this work, the known results on approximation algorithms for packing problems were either highly fixed for one problem or restricted to one class of objects, mainly polygons and hypercubes.
    In this sense, our framework is the first result with a general toolbox flavor in the context of approximation algorithms for geometric packing problems. Thus, we believe that our technique is of independent interest, being possible to inspire further work on geometric packing.
\end{abstract}

\section{Introduction}
\label{sec:introduction} 

In packing problems, we have a set of items that must be packed in one or more containers, called bins, optimizing some resource. In geometric packing problems, the items and the bins are geometric objects, e.g., (hyper)cubes, (hyper)rectangles, (hyper)spheres, etc. A packing is a non-overlapping arrangement of the items within the bins, and the objective function can vary from minimizing the number or the size of the bins to maximizing the profit associated with the packed items.

Geometric packing problems are classic and relevant problems that have been studied in mathematics for centuries. For instance, in the $17$th century, \citet{kepler-1611_conjecture-spheres-density} conjectured a bound on the average density of any packing of congruent spheres in the Euclidean space. It was only in $2006$, after centuries, that \citet{HalesFerguson2006-kepler-conjecture-proof} presented a formal proof in the affirmative. Then it took more that ten years until further achievements were presented. In $2017$, the Fields Medal winner \citet{Viazovska2017-sphere-packing-8d} gave an optimal packing of equal spheres in the $8$-dimensional space, and together with other authors (\citet*{CohnEtal2017-sphere-packing-24d}), extended the result to $24$ dimensions.

Regarding computational complexity, several geometric packing problems are \classNPH{}~\cite{BermanEtal-1981,DemaineEtal2010c,Fowler-1981,kimMiltzow-2015,LeungEtal-1990}.
Therefore it is natural to approach these problems with techniques that may compromise optimality. In fact, there are many heuristics and exact algorithms for the problem of maximizing the packing density~\cite{RomanovaEtal_circles-in-circles_2023,Amore_circle-in-polygon_2023,FuEtal_sphere-in-cilinder_2016,BirginLobato_ellipsoids-packing_heuristic_2019,HifiLabib_sphere-in-3dregions_heuristics_2019}, as well as for the problem of minimizing the size of the container~\cite{BirginSobral_circle-in-sphere_2008,BirginEtal_circles-in-ellipses_2013,BirginEtal_ellipsoids_2015,ZengEtal_circles-in-circle_2016,AkebEtal_circles-in-circles_2009}.
We refer the reader to the survey of \citet{HifiMhallah_survey_sphere-packing-2009}.
In the context of approximation algorithms, however, the literature is not so vast and most of the results regard rectangular shapes and $d$-dimensional boxes.
In the \emph{bin packing problem}, the goal is to pack all items into the minimum number of bins. For rectangular items and bins, the best known result is an asymptotic $1.405$-approximation due to \citet{BansalKhan2014_1.405-approx-2d-bin-packing}. There is also an {\aptas} for the hypercube bin packing problem, given by \citet{BansalEtal_APTAS-d-dimensional-bin-packing_2006}. 
A closely related problem is the \emph{strip packing problem}, where the goal is to pack all items into a bin of fixed width and minimum height. For the case with rectangular items and strip, \citet{KenyonRemila_APTAS-2d-strip-packing_2000} gave an {\aptas}. We refer the reader to the works of \citet{ChristensenEtal_survey-bin-packing_2017} and \citet{CoffmanEtal_survey-bin-packing_2013} for an extensive review on the bin and strip packing problems. 
Another classic problem is the \emph{geometric knapsack problem}, where the items are associated with profits and the objective is to pack a subset of them in a bin (knapsack) maximizing the total profit of the packed items.
For rectangular items and knapsack, the best polynomial-time approximation algorithm is a $1.89$-approximation due to \citet{GalvezEtal_knapsack-rectangles-2021}, while in the pseudo-polynomial setting, \citet{GalvezEtal-2021_rectangle-knapsack_4/3+eps-approximation} gave a $(4/3+\eps)$-approximation algorithm. 
Still in two dimensions, \citet{MerinoWiese_knapsack-for-convex-polygons-2020} gave quasi-polynomial-time approximation algorithms for the version with convex polygons as items and a squared knapsack, assuming polynomially bounded integral input data.
In higher dimensions, \citet{JansenEtal_PTAS-hypercubes-knapsack_2022} gave a {\ptas} for the version where the items and the knapsack are restricted to hypercubes. 

\subsection{Our contribution}

In this work we concern the study of geometric packing problems through the lens of approximation algorithms.
We consider convex fat objects, 
i.e., the ratio between their largest and smallest dimensions is bounded by a fixed constant.
This class of objects includes many common shapes, such as hyperspheres, ellipsoids, hypercubes and regular polytopes.
When dealing with general objects, some numerical challenges may arise, one of which involves the possible presence of irrational numbers.
In fact, for non-rectilinear shapes, such as circles, it remains an open question whether it is always possible to pack a set of objects into a bin using only rational coordinates, even when all input data is rational.
One strategy to address this issue is rounding the potentially irrational coordinates in a feasible packing to achieve a packing where all coordinates are rational.
This strategy, in turn, leads to another issue: ensure that the rational packing is also feasible. 
A common approach towards this issue is called \textit{resource augmentation}. In the context of geometric packing, a bin is said to be \textit{augmented} if its dimensions are slightly increased.
The augmentation allows a rearrangement of the items to rational coordinates, ensuring a feasible rational packing in an augmented bin.
When designing approximation schemes, two types of precision parameters may be present: the approximation ratio itself, identified by $\eps$ in standard notation, 
and the ratio between the sizes of the original and augmented bins, which we refer to as \emph{augmentation ratio}.
Depending on the specific problem and shapes at hand, our approximation schemes may involve also the resource augmentation ratio, rather than just the standard approximation ratio.

If a scheme has only the approximation ratio, then it is a standard \ptas{}.
If a scheme has only the augmentation ratio, we call it \emph{resource augmentation scheme} and write \ras{}. Finally, if a scheme has both the approximation and the augmentation ratios, we call it an \emph{augmented \ptas{}} and write \augptas{}. 
An approximation scheme is said to be \textit{efficient} if the exponent of the polynomial terms in its time complexity does not increase as $\eps$ shrinks; equivalently, it is \textrm{FPT (Fixed-parameter Tractable)} algorithm parameterized on $\eps$.
Formally, the running time is of the form $\bigO(n^c f(1/\eps))$, where $c$ is a constant independent of $\eps$ and $f$ is a computable function.
We denote by \eaugptas{}, \eras{} and \eptas{} the efficient versions of the aforementioned approximation schemes.

Our result is a framework that yields different types of efficient approximation schemes for several packing problems.
Numerous packing problems have been extensively studied in the literature. 
Below, we define the main problems addressed by our framework in this work.

\begin{problem}[Multiple Knapsack Problem - \textrm{MKP}]
\label{prob:mkp}
    Given a set of items $\cali$, each with an associated profit, and $m$ hypercuboidal knapsacks, pack a subset of $\cali$ into the $m$ knapsacks such that the total profit of the packed items is maximized.
\end{problem}

\begin{problem}[Bin Packing Problem - \textrm{BPP}]
\label{prob:bp}
    Given a set of items $\cali$ and an unlimited amount of hypercuboidal bins, pack all objects of $\cali$ into the minimum number of bins.
\end{problem}

\begin{problem}[Multiple Strip Packing Problem - \textrm{MSPP}]
\label{prob:msp}
    Given a set of $d$-dimensional items $\cali$, $d-1$ lengths $\range{l_1}{l_{d-1}}$, and a number $m$, find the minimum value $h$ such that $\cali$ can be packed into $m$ hypercuboidal bins of size $l_1 \times \dots \times l_{d-1} \times h$.
\end{problem}

\begin{problem}[Multiple Minimum Container Problem - \textrm{MMCP}]
\label{prob:mmc}
    Given a set of items $\cali$ and a number $m$, find the minimum value $l$ such that $\cali$ can be packed into $m$ hypercubes of side length $l$.
\end{problem}

In the following, we list the results we explicitly derived in this work. 
Nonetheless, we believe the framework has potential to be applied to other settings, therefore being of independent interest.
Unless otherwise stated, all the problems consider $d$-dimensional convex fat objects as items.

\begin{enumerate}
    \item \eras{} and \eaugptas{} for the Multiple Knapsack problem;
    \item \ptas{} for the Multiple Knapsack problem of hyperspheres;
    \item \eras{} for the Bin Packing problem;
    \item \eptas{} for the Multiple Strip Packing problem for a wide class of geometric objects;
    \item \eptas{} for the Multiple Minimum Container problem;
    \item All the schemes allowing item multiplicity;
    \item All the resource augmentation schemes allowing item rotation.
\end{enumerate}

All the resource augmentation schemes for fat objects also allows rotation on the objects by arbitrary angles. This result, in particular, brings another extra flavor to our framework, since most results comprising such general objects are limited to translations. 
We highlight the fact that our framework easily handles item multiplicity, which is a nice feature of our method, as the introduction of item multiplicity adds considerable complexity to most problems. 
In the geometric settings, all that is known for most problems is that they belong to \textrm{EXPSPACE}.
A notable example is the pallet loading problem.

Additionally, in knapsack versions of packing problems, we are able to handle additional conditions on the items.
This includes very common and pertinent constraints, such as:

\begin{itemize}
    \item \textit{Conflict constraints}: some pairs of items cannot be packed together;
    \item \textit{Multiple-choice constraints}: Given a subset $F$ of items, at most one of them can be selected to the solution;
    \item \textit{Multiple capacity constraints}: Given a set of weights for each item and corresponding capacities to the knapsack, the sum of the weights of the packed items cannot exceed its corresponding capacity.
\end{itemize}

We observe that these constraints are tractable if the number of constraints and associated items are bounded by a constant.
Despite the restrictive appearance of such conditions, it is not expected to be able to handle an asymptotically larger number of constraints or associated items.
To illustrate this, if we relax the condition of having a constant number of associated variables, then we can model the vector multidimensional knapsack problem, which does not admit an \eptas{} even with only $2$-dimensional vectors unless $\classP = \classNP$.
Moreover, if we relax the condition of having a constant number of constraints, then we are able to formulate the independent set problem, which does not even admit a $(1/n^{1 - \eps})$-approximation for any $\eps > 0$ unless $\classP = \classNP$.

Considering the versatility and robustness of our method, we regard it as a framework for deriving good approximation results across a substantial spectrum of classic packing and other related problems. For that reason, we believe it holds independent interest. We also believe it can be further explored even to other classes of optimization problems that share some similarities with packing problems, such as scheduling problems.
A prominent feature of our technique is the ease of incorporating modifications to accommodate additional constraints. While this can also be done using other methods, like dynamic programming, our approach provides a more straightforward means of implementation. A prime example is the introduction of item multiplicity: whereas other methods might require specific ad-hoc adjustments, our model requires no alterations.

\subsection{Our technique}

The packings produced by our framework follow a certain structure.
The idea behind such structure is build upon a strategy shown in the work of \citet{MiyazawaEtal2015a}, which is based on partitioning the circles in a way that, after discarding or separating one subset of circles with negligible total volume (\emph{medium items}), the remaining circles (\emph{\nonmedium{} items}) are organized in levels, where circles of a level are much smaller than those of previous levels. 
Then each level can be handled independently, using bins whose size is proportional to the radii of its circles, before all levels are merged to obtain a packing of the whole instance.
Having a small volume, the medium items can also be packed using small bin volume and can be combined more easily with the remaining packing. 
This partitioning strategy gives flexibility to calibrate the gap of radii among circles of two consecutive levels, as well as the gap of the size between circles and bins within the same level. 

Our resource augmentation scheme also starts with the partitioning of the circles into medium ($H_t$) and \nonmedium{} ($S_0, S_1, \ldots$) items, such that the medium items have negligible volume. The algorithm computes a super-optimal packing for each part and combine them to obtain a super-optimal packing of the input items.
To obtain a super-optimal packing of the medium items, we make use of the classic \nfdh{} algorithm, as well as its generalization for higher dimensions.
To obtain a super-optimal packing of the \nonmedium{} items, we use a configuration-based IP with a special feature: it can be decomposed into independent blocks, where each block concerns the items of one level.
The one aspect that relates two consecutive levels (and therefore two consecutive blocks) is that the packing of a level determines the amount of free volume left available for packings of the next level.
However, opposed to the bin packing problem, a simple greedy choice of maximizing the area occupation in each level is not suitable for the knapsack problem. For instance, if the instance has small circles of very high profit, it is probably a better choice to restrict the area used by large circles in order to reserve more space of the knapsack for the highly profitable small circles.

The idea with the configuration-based IP is to solve its linear relaxation and round up the fractional solution.
This leads to extras bins that must be accommodated somewhere. To have a good bound on this number, we need a good bound on the number of non-null variables in each level.
At first sight, we do not have any control on how the non-null variables are distributed among levels. In this point, the fact that our model can be decomposed in independent blocks is crucial. 
From a first solution of the whole model, we can easily obtain solutions from each level independently. This nice feature of our IP allow us to guarantee the achievement of a \emph{balanced solution}, which is exactly a solution where the non-null variables are well distributed among levels, which in turn gives us a control over the number of extra bins.
One point of attention is that the first level represents the knapsack itself, therefore extra bins of this level is not affordable. 
To handle this, we actually solve a \textrm{MILP}, where the variables corresponding to configurations of the first level are kept integer. Since we have a constant bound on the number of configuration in each level, the \textrm{MILP} can be solved in polynomial time with, for example, the algorithm given by \citet{Lenstra1983}.

To derive a \ptas{} for hyperspheres, we take good advantage of topological properties to guarantee that, for any packing of the circles in the first level (that go in the knapsack itself), a big enough amount of volume of the knapsack is left free to be used on the next level.
Given the huge difference in size between two consecutive levels, this is enough to guarantee that, possibly by discarding a subset with negligible profit, we can accommodate the extra bins back in the knapsack, \emph{without} the need of resource augmentation.
We call attention to the fact that the same strategy could be applied to other geometric objects. In fact, although we illustrate the idea with hyperspheres, for our framework to derive a \ptas{} for other objects, it suffices to have such guarantee of unused volume in any packing of the first level. No adjustments are needed in the other parts of the algorithm.

Another aspect of our \ptas{} for hyperspheres that must be mentioned is how we obtain the realization itself of a packing.
The configurations give a description of a packing, that is, how many items of each size are presented in the packing.
To have a realization out of such description, we need to obtain the actual packings corresponding to each configuration. 
In this moment, we must handle the issue that a packing of a set of circles (hyperspheres) may require irrational coordinates. 
We use an algorithm that, given a set of circles (hyperspheres) of rational radii and a bin of rational side lengths, it decides if the circles (hyperspheres) can be packed in the bin and, in the affirmative, it produces a rational packing with some precision errors. In practice, this means that the rational packing produced may present small overlaps among circles and/or the sides of the bin. 
To make such packing feasible, we use the following strategy: each circle (hypersphere) is slightly shifted until there is no more overlaps; this shifting causes the need of an augmented bin. We push the circles (hyperspheres) upwards to ensure that the augmentation is restricted to only one dimensional (the height of the bin).
We highlight that, if it was proved that there is a way to obtain, in polynomial time, a rational realization of a packing, our algorithm would not require resource augmentation. 
It would suffice such hypothetical method to obtain the realization of each configuration.

\subsection{Other related work}

The investigation of hypersphere packing within the realm of approximation algorithms is relatively new. In fact, in the context of combinatorial optimization, the first approximation schemes emerged in $2016$ with the work of \citet{MiyazawaEtal2015a} on hypersphere bin packing and strip packing problems. They gave \aptas{}s for both problems.
{This study presented some innovative approaches, notably a structural theorem that relies on recursively subdividing spheres by their radii to derive a hierarchical decomposition of a packing in levels; and an algebraic method to determine the feasibility of packing a collection of hyperspheres into a specified bin.
Although partitioning the input set of items by size has become a standard method in packing problems, the recursive division they proposed represented a novel approach.
Furthermore, although the application of algebraic methods to sphere packing had been suggested, the idea of obtaining a packing from quantifier formula algorithms was proposed in their work. This included clever strategies to address numerical precision challenges inherent in realizing a packing of hyperspheres, particularly due to the potential requirement for irrational numbers. Indeed, they first posed the question of whether a set of hyperspheres with rational radii can be packed using exclusively rational coordinates.}

Later, in $2018$, \citet{LintzmayerEtal2018p} gave a \ptas{} for the circle knapsack problem in the special case where the circles' profits are their respective areas. In this sense, the problem becomes that of optimizing the density of the packing. 
\citet{MiyazawaWakabayashi2022_survey-circle-packing} present a review of some techniques for circle and hypersphere packing, and also a $(1/3-\epsilon)$-approximation algorithm for the weighted circle knapsack problem.

With our work, we advance the literature on the hypersphere knapsack problem with general profits. 

This manuscript is an extended version of the work presented in WAOA'23~\cite{ChagasEtal-2023_approx-scheme-hypersphere-knapsack}, where we presented a first version of our framework by means of a \augptas{} for the hypersphere multiple knapsack problem. We highlight however that the core ideas of all of our results were already registered in that conference paper. 

Subsequent to that work, \citet{AcharyaEtal-2024_ptas-sphere-knapsack} proposed a \ptas{} for the hypersphere knapsack problem. The structural ideas and techniques they use are closely related to those introduced in the work of \citet{MiyazawaEtal2015a} in $2016$. 
Their main difference from our work is the use of a dynamic programming approach to obtain a solution. 
We argue that our nice IP model presents as a innovative and more robust approach, certified, for instance, by the ease of incorporating modifications to accommodate additional constraints. While this can also be done by using other methods, like dynamic programming, our approach provides a more straightforward means of implementation. A prime example is the introduction of item multiplicity: whereas other methods might necessitate specific ad-hoc adjustments, our model requires no alterations.
Another contrast is in time complexity, our framework is efficient, yielding \eptas{}, \eras{} and \eaugptas{}.

\subsection{Organization of the text}

First, in \cref{sec:preliminaries}, we review key results from the literature that are relevant to this work, including structural lemmas for circle packing and the \nfdh{} algorithm.
To facilitate the understanding of our framework, we initially present it in \cref{sec:ckp-main} using the circle multiple knapsack problem as a simpler illustrative example.
Then, in subsequent sections, we progressively generalize our framework. Specifically:
\begin{itemize}
    \item In \cref{sec:other-packing-problems}, we present our approach to other packing problems.
    \item In \cref{sec:generalization-fat-objects}, we demonstrate that our framework actually deals with convex fat objects, rather than just circles or hyperspheres.
    \item In \cref{sec:additional-features}, we introduce additional problem features that our framework can handle.
\end{itemize}
Finally, in \cref{sec:final-remarks}, we provide concluding remarks.
To streamline the main exposition, the proofs of most auxiliary lemmas and minor theorems are deferred to the appendix.

\section{Preliminaries}
\label{sec:preliminaries}

We assume that all objects lie in the Euclidean space.
If $p$ and $q$ are two points in the plane, their Euclidean distance is denoted by $\dist{p}{q}$. 
Given a set $\cals = \set{s_1,\ldots,s_n}$ of $n$ circles, we denote the radius and the diameter of each circle~$s_i \in \cals$ by $r_i$ and $d_i$, respectively.
For a rectangle~$B$ of rational width $w$ and height $h$, we write $B_{w \times h}$ and we call $w \times h$ the \emph{size} of $B$. When the context is clear, we may omit the size from the notation.
For a two-dimensional geometric object $D$ we denote its area by $\area{D}$, and if~$D$ is a set of objects, then $\area{D} = \sum_{A \in D} \area{A}$.
When no ambiguity arises, we denote the area of a circle of radius $r$ simply by $\area{r}$.
When dealing with a more general $d$-dimensional object $D$, we denote by $\vol{D}$ its volume and $\surf{D}$ the area of its surface.
In the same manner as the area, these notations are also used for a set of objects.
Given a positive integer~$n$, we write $[n] = \srange{1}{n}$.

\subsection{Circle Packing and Gap-Structured Partition}
\label{sec:circle-packing}

A packing of a set of circles into a bin is an attribution of coordinates to the center position of each circle such that no two circles overlap and each circle is entirely contained in the bin.
The study of circle packing in the lens of approximation algorithms began with the work of \citet{MiyazawaEtal2015a} in \citeyear{MiyazawaEtal2015a}. 
Their work introduced concepts and techniques that served as a baseline for subsequent results in the field.
We summarize some of their results in this section.
They investigated the \textit{circle bin packing problem} (\cbp{}), which is \cref{prob:bp} with items restricted to $2$-dimensional circles, that is: given a set $\cali$ of circles with rational diameters and values $w, h \in \Q_+$, pack all circles of $\cali$ into the minimum number of bins of size $w \times h$.
We denote an instance of the \cbp{} by $(\cali, w, h)$ and its optimal solution by $\optBP{\cali}{w}{h}$.

Circle packing problems raise an intrinsic issue: It is not known if, for every rational instance of the problem, there always exists an optimal solution where the center of every circle is given by rational coordinates.
For the \cbp{}, \citet{MiyazawaEtal2015a} handle this issue providing
an algorithm that always produces rational solutions, but in augmented bins.
Briefly, the idea is to formulate the problem as a system of polynomial inequalities, where the variables correspond to the center position of the circles; then the decision problem of whether a set of circles can be packed in a bin is equivalent to deciding if such system admits real solutions.
The set of solutions that satisfy the system is a semi-algebraic set in the field of the real numbers, therefore any algorithm for the more general quantifier elimination problem can be used to decide whether such set is empty.
With this strategy, we can decide whether a set of circles fit into a bin without overlaps, even if irrational coordinates were to be necessary.
However, we cannot guarantee a realization of such packing in rational coordinates, since the positions are given by roots of polynomials, which may be irrational numbers.
Trying to adjust them to rational coordinates may introduce overlaps with the borders of the bin or among circles, resulting in an approximate packing, which is defined as follows.
For some number $\xi$, an attribution of a set $\cali$ of circles in coordinates $p_i = (x_i, y_i)$, for each $s_i \in \cali$, into a bin of size $w \times h$ is a \emph{$\xi$-packing} if no two circles overlap by more than $\xi$ and no circle overlaps the bin by more than $\xi$ in any dimension. Formally, in a $\xi$-packing it holds the following inequalities.
\begin{alignat*}{2}
    \dist{p_i}{p_j} &\geq r_i + r_j - \xi \geq 0 & \quad & \forall\, s_i, s_j \in \cali, s_i \neq s_j, \\
    r_i - \xi &\leq x_i \leq w - r_i + \xi       &       & \forall\, s_i \in \cali, \\
    r_i - \xi &\leq y_i \leq h - r_i + \xi       &       & \forall\, s_i \in \cali.
\end{alignat*}

\citet{MiyazawaEtal2015a} present a shifting strategy to rearrange the circles within the bin until there is no overlap, resulting in an increase in the height of the bin by a small value.
We state this result in the next lemma.

\begin{lemma}[\citet{MiyazawaEtal2015a}]
\label{thm:shifting-algorithm}
    Given a set $\cali$ of circles, with $\size{\cali} = n$, and an $\eps h$-packing of $\cali$ into a bin $B_{w \times h}$, for some $\eps> 0$, we can ﬁnd a packing of $\cali$ into a bin of size $w \times (1 + n\sqrt{6\eps})h$ in linear time.
\end{lemma}
 
Using this result combined with an algorithm similar to that of \citet{fernandez1981} for the one dimensional bin packing (the breakthrough linear grouping strategy), \citet{MiyazawaEtal2015a} obtain a super-optimal solution for the \cbp{} in augmented bins, for the case where the radius of each circle is at least a given constant. The following lemma gives a formal statement of such result.

\begin{lemma}[\citet{MiyazawaEtal2015a}]
\label{thm:fkm-etal_bin-packing-in-augmented-bins}
    Let $(\cali, w, h)$ be an instance of the circle bin packing problem, with $w, h \in \bigO(1)$, $|\cali| = n$, $\min_{1 \leq i \leq n} r_i \geq \delta$ and $|\srange{r_1}{r_n}| \leq K$, for constants $\delta$ and $K$. 
    Given a number $\gamma > 0$, there exists an algorithm that produces a packing of $\cali$ into at most $\optBP{\cali}{w}{h}$ bins of size $w \times (1+\gamma)h$, in polynomial time on $n$.
\end{lemma}

In possession of this result, they derive an \aptas{} under resource augmentation for \cbp{} problem.
The key idea is to prove the existence of a near optimal packing that respects a well-behaved structure and wastes little area compared to an optimal solution.
Such structure relies on a fastidious partitioning of the instance, which is called a \textit{gap-structured partition}, explained in the following.

Let $(\cali, w, h)$ be an instance of the \cbp{} and let $\eps > 0$ be a constant.
We define $\reveps = 1/\eps$.
We partition $\cali$ into groups $G_i = \set{ s_j \in \cali : \eps^{2i} w \geq d_j > \eps^{2(i+1)} w }$, for $i \geq 0$. 
Then we partition these groups into sets $H_{\ell} = \bigcup_{i \equiv \ell \pmod{\reveps}}{G_i}$, for $0 \leq \ell < \reveps$.
Now consider a fixed set $H_t$ for some index $0 \leq t < r$. We refer to $H_t$ as the \textit{medium items}.
By removing the set $H_t$ from the instance, we can arrange the remaining groups into sets of groups such that there is a significant gap on the radii of circles of any two consecutive sets.
For that purpose, we define sets $S_j = \bigcup_{i = t + (j-1)\reveps + 1}^{t + j\reveps - 1} G_i$, for $j \geq 0$. 
See \cref{fig:instance-partitioning} for an illustrative sketch.
We denote by $\struct{\cali}{t} = \cali \setminus H_t = \bigcup_{j \geq 0} S_j$ the \textit{level items} and say that $H_t, S_0, S_1, \dots$ is a \emph{gap-structured partition} of $\cali$.
The minimum and maximum radii of $S_j$ are denoted by $\rmin{j}$ and $\rmax{j}$, respectively. 
The strategy is to pack each $S_j$ into bins of appropriate size $w_j\times h_j$ according to the radii of the circles.
We set $w_0 = w$, $h_0 = h$, representing the knapsack itself, and for $j \geq 1$, we set $w_j = h_j = \eps^{2(t + (j-1)\reveps) +1} w$. 
We say that a \emph{grid} of cell \emph{size} $w_j \times h_j$ over a bin $B$ divides $B$ into a set $\grid{j}{B}$ of  \emph{cells} of size $w_j \times h_j$.
The empty cells of this grid (cells that do not intersect circles of previous levels) can be used as bins to pack items of $S_j$.
Additionally, we say that a bin 
$B'$ of size $w_j \times h_j$ \textit{respects} $w \times h$ if $B' \in \grid{j}{B}$.
To avoid verbosity, hereafter we refer to each $j$ as \emph{level}~$j$, and to circles of $S_j$ and bins of size $w_j \times h_j$ simply as circles and bins of level~$j$. 

\begin{figure}[htb!]
    \centering
    \input{instance-partitioning}
    \caption{Sketch to illustrate the partitioning of the original instance.}
    \label{fig:instance-partitioning}
\end{figure}

We highlight two important properties regarding the sets $S_j$, for $j \geq 1$: \
$i)$ within the same level, circles are small compared to bins, but not too small, as there is a constant bounding the maximum number of circles of $S_j$ into bins of size $w_j\times h_j$; and \
$ii)$ between two consecutive levels $j$ and $j+1$, circles and bins of level~$j+1$ are much smaller than circles and bins of level~$j$.
This indicates that after packing circles of a level in the corresponding bins for that level, the area left unoccupied may accommodate a large number of circles (and bins) of the subsequent level. 
The idea is to recursively use grids to build a packing respecting the following structure: For each level~$j$, circles are packed in bins of their respective levels, over which it is drawn a grid of size $w_{j+1} \times h_{j+1}$; the empty cells of this grid are then used to pack circles of $S_{j+1}$, as illustrated in \cref{fig:structured-packing}. 
For clarity, from level~$1$ onward, we say subbins instead of bins.
In the following, we present a formal definition.

\begin{figure}[tb!]
    \centering
    \begin{tikzpicture}[
        cir/.style={draw, circle, minimum width=#1 cm},
        cirS2/.style={draw=black,fill=white},
        gridS2/.style={step=0.25, black!80, ultra thin}
    ]
    \tikzset{
        circlesS0/.pic={
            \node[cir=2.5] at (0.0, 0) {};
            \node[cir=2.15] at (2.45, -0.2) {};
        },
        circlesS1/.pic={
            \node[cir=0.5] at ( 0.0, 0.0) {};
            \node[cir=0.4] at ( 0.15, 0.5) {};
            \node[cir=0.4] at ( 0.35, -0.4) {};
            \node[cir=0.4] at ( 0.55, 0.1) {};
            \node[cir=0.4] at ( 0.7, 0.55) {};
            \node[cir=0.5] at ( 1.1, 0.1) {};
            \node[cir=0.4] at ( 0.85, -0.35) {};
            \node[cir=0.5] at ( 1.5, 0.5) {};
            \node[cir=0.4] at ( 1.4, -0.3) {};
            \node[cir=0.5] at ( 1.8, 0) {};
            \node[cir=0.4] at ( 2.1, 0.5) {};
            \node[cir=0.4] at ( 2, -0.5) {};
            \node[cir=0.4] at ( 2.35, -0.05) {};
        },
        circlesS2/.pic={
            \foreach \i in {0, ..., 6} {
                \fill[shift={(0.45*\i, 0)}, rotate=45*\i, transform shape]  pic {circlesS2set};
                \fill[shift={(0.4 + 0.45*\i, -0.4)}, rotate=45*\i, transform shape]  pic {circlesS2set};
            }
        },
        circlesS2set/.pic={
            \fill[cirS2] ( 0.00, 0.00) circle [radius=0.04cm];
            \fill[cirS2] (-0.10, 0.15) circle [radius=0.05cm];
            \fill[cirS2] ( 0.12, 0.10) circle [radius=0.035cm];
            \fill[cirS2] ( 0.10,-0.12) circle [radius=0.05cm];
            \fill[cirS2] (-0.11,-0.11) circle [radius=0.04cm];
        },
        packingS2/.pic={
            \fill[cirS2] (0.05,0.05) + (-0.125,-0.125) circle [radius=0.04cm];
            \fill[cirS2] (0.1,0.125) + (-0.125,-0.125) circle [radius=0.04cm];
            \fill[cirS2] (0.15,0.05) + (-0.125,-0.125) circle [radius=0.03cm];
            \fill[cirS2] (0.03,0.15) + (-0.125,-0.125) circle [radius=0.03cm];
            \fill[cirS2] (0.075,0.2) + (-0.125,-0.125) circle [radius=0.03cm];
            \fill[cirS2] (0.16,0.2) + (-0.125,-0.125) circle [radius=0.035cm];
            \fill[cirS2] (0.175,0.125) + (-0.125,-0.125) circle [radius=0.03cm];
        }
    }
    \begin{scope}
    
        \node at (-2,9) {$S_0$};
        \fill (0, 9) pic {circlesS0};
    
        \node at (-2,6.4) {$S_1$};
        \fill (0, 6.4) pic {circlesS1};
    
        \node at (-2,4.8) {$S_2$};
        \fill (-0.3, 5) pic {circlesS2};
    
        
        \node[cir=2.5] at (6.25, 6.25) {};
        \node[cir=2.15] at (7.9, 7.9) {};
    
    
        \node[cir=0.5] at (5.25, 8.25) {};
        \node[cir=0.4] at (5.725, 8.2) {};
        \node[cir=0.4] at (5.2, 8.7) {};
        \node[cir=0.5] at (5.65, 8.65) {};
    
        \node[cir=0.5] at (8.25, 5.25) {};
        \node[cir=0.5] at (8.25, 5.75) {};
    
    
        \foreach \i in {8.50, 8.75} {
            \foreach \j in {5.00, 5.25, ..., 5.75} {
                \fill (\i, \j) + (0.125, 0.125) pic {packingS2};
            }
        }
    
        \draw[step=1] (5,5) grid (9,9);
        \draw[gridS2] (8,5) grid (9,6);
        \draw[gridS2] (5,8) grid (6,9);
        
    \end{scope}
    
    \end{tikzpicture}
    \caption{Illustration of a structured packing of the level items.}
    \label{fig:structured-packing}
\end{figure}

\begin{definition}
Consider a set $\cali$ of circles. We say that a packing of the level circles, $\struct{\cali}{t}$, in a bin $B_{w \times h}$ is a \emph{structured packing} if the following holds:
\begin{itemize}
    \item $S_0$ is packed in $B$;
    \item for every $j \geq 1$, $S_j$ is packed in a subset $D_j \subseteq \grid{j}{B'}$ of subbins of size $w_j \times h_j$, where $B'$ is a bin or subbin of level $j-1$; and
    \item for every subbin $D' \in D_j$, $D'$ does not intersect any circle from $S_{\ell}$, for $\ell < j$.
\end{itemize}
\end{definition}

Note that in a structured packing, the subbins that partially intersect a circle of some previous level are not used to pack circles of subsequent levels, even thought they have some unoccupied area. This causes a waste of used area compared to that of an optimal solution. However, as it is stated in the next lemma, such waste of area is small.

\begin{restatable}[\citet{MiyazawaEtal2015a}]{lemma}{thmWastedArea}
\label{thm:wasted-area}
    Let $A \subseteq S_j$ be a set of circles packed in a bin $B_{w_j \times h_j}$ and $D \subseteq \grid{j+1}{B}$ be the subset of grid cells of size $w_{j+1} \times h_{j+1}$ intersecting but not entirely contained in circles of $A$. Then $\area{D} \leq 16\eps \area{A}$.
\end{restatable}

\citet{MiyazawaEtal2015a} show that for any instance $(\cali,w,h)$ of the \cbp{} and a gap-structured partition of $\cali$, there is a structured packing of $\struct{\cali}{t}$ using only a small amount of extra area.

\begin{lemma}[\citet{MiyazawaEtal2015a}] \label{thm:structured-bin-packing}
    Let $(\cali, w, h)$ be an instance of the \cbp{} and let $H,S_0,S_1,\ldots$ be a gap-structured partition of $\cali$.
    There exists a structured packing of $\cali\setminus H$ into a set of bins $D$ that respect $w \times h$, and such that $\area{D} \leq (1 + 44\eps) \opt_{w \times h}^{\bp}(\cali) wh$. 
\end{lemma}

In one step of our framework, we employ the ideas presented in this subsection, so we refer to them again further in the text.

\subsection{The Next Fit Decreasing Height Algorithm}
\label{sec:nfdh}

In this section, we discuss a generalization of the classic \nfdh{} algorithm to higher dimensions. Further on, we use these ideas in the packing of the medium items.

The \textit{next fit decreasing height} (\nfdh{}) procedure is an algorithm originally named for the two-dimensional strip packing problem of rectangles and introduced by \citet{Coffman1980}, but the algorithm had appeared and analysed before.
It consists in sorting the rectangles in non-increasing order of height and packing them in a shelf-like manner: 
Starting from the bottom left of the bin, the rectangles are sequentially positioned adjacent to one another until the subsequent item would overlap the right border of the bin.
At this point, the algorithm defines a horizontal line, parallel to the bottom of the strip, at the top of the highest rectangle packed so far; such line is then used as if the new bottom of the strip and the packing proceeds in the same manner, i.e., packing the remaining rectangles side by side, in non-increasing order. The algorithm stops when all rectangles are packed. In this sense, the rectangles are packed in levels.
\cref{fig:nfdh} illustrates a packing of squares produced by the \nfdh{}.

\begin{figure}[htb!]
    \centering
    \scalebox{0.65}{
\begin{tikzpicture}
\tikzset{
pics/square/.style={
    code = {
        \begin{scope}
            \draw[fill=gray!20] (0,0) rectangle (#1, #1);
            \node (-left)  at ( 0,0) {};
            \node (-right) at (#1,0) {};
        \end{scope}
    }
},
pics/listsquare/.style={
    code = {

        \readlist \squares {#1}
        \edef \n {\listlen\squares[]}
        
        \pic (c1) at (0,0) { square=\squares[1] };
        \foreach \i in {2, ..., \n} {
            \edef \length {\squares[\i]}
            \pgfmathsetmacro{\j}{int(\i - 1)}
            \pic (c\i) at ($(c\j-right)$) {square=\length};
        }

        \node (-end) at (c\n-right |- 0,0) {};
    }
}
}

    \draw (0,6) -- (0,0) -- (9.4,0) -- (9.4,6);

    \pic at (0,0)   { listsquare={2, 2, 1.9, 1.75, 1.6} };
    \pic at (0,2)   { listsquare={1.5, 1.4, 1.3, 1.3, 1.2, 1.2, 1.2} };
    \pic at (0,3.5) { listsquare={1.2, 1.1, 1, 0.9, 0.8, 0.8, 0.7, 0.7, 0.65, 0.6, 0.6} };
    \pic at (0,4.7) { listsquare={0.55, 0.5, 0.5, 0.45, 0.45, 0.45, 0.4, 0.4, 0.35, 0.35, 0.35, 0.35, 0.3, 0.3, 0.3, 0.25, 0.25, 0.25, 0.2, 0.2, 0.2, 0.15, 0.15, 0.1, 0.1, 0.1} };

\end{tikzpicture}
}
    \caption{\nfdh{} applied for a set of squares.}
    \label{fig:nfdh}
\end{figure}

Due to its guarantees on the packing density, the \nfdh{} algorithm is used in several packing problems to obtain approximation algorithms.
Since its introduction, generalizations of \nfdh{} had been extensively investigated to higher dimensions and other packing variants, with the first analyzes dating back to \citeyear{Meir1968}~\cite{Meir1968}.
The $d$-dimensional \nfdh{} algorithm for hypercubes can be explained in a recursive manner.
Electing one dimension as the height, it starts with a base on the bottom of the bin in this dimension and considers the projection of the hypercubes and the bin in the remaining $d-1$ dimensions. Then it packs the largest amount of hypercubes using the $(d-1)$-dimensional version of \nfdh{}. At last, it packs such items in the base and shifts the base to the top of the largest packed hypercube, repeating the process for the remaining items.

\citet{Meir1968} showed a volume-based sufficient condition so that the \nfdh{} algorithm packs a set of hypercubes in a bin.

\begin{theorem}[\citet{Meir1968}]
    Let $\cali$ be a set of $d$-dimensional hypercubes, with side lengths at most $\delta$. The \nfdh{} algorithm can pack $\cali$ in an rectangular hypercuboidal bin of size $\ell_1 \times \ell_2 \times \dots \times \ell_d$ if $\ell_i \geq \delta$, for $i\in[d]$, and
    \begin{displaymath}
        \vol{\cali} \leq \delta^d + \prod_{i=1}^d(\ell_i - \delta).
    \end{displaymath}
\end{theorem}

By isolating $\ell_d$ in this theorem, we obtain the following result.

\begin{corollary} \label{thm:nfdh-height}
    Let $\cali$ be a set of $d$-dimensional hypercubes, with side lengths at most~$\delta$. The \nfdh{} algorithm can pack $\cali$ in a rectangular hypercuboidal bin of size $\ell_1 \times \ell_2 \times \dots \times \ell_d$ if
    $\ell_i\geq \delta$, for $i\in[d]$, and
    \begin{displaymath}
        \ell_d \geq \frac{ \vol{\cali} - \delta^d }{ \prod_{i=1}^{d-1}(\ell_i - \delta) } + \delta.
    \end{displaymath}
\end{corollary}

\citet{Harren2009} showed an efficiency guarantee on the occupied volume by the \nfdh{} algorithm based on the side lengths of the hypercubes and surface area of the bin.

\begin{theorem}[\citet{Harren2009}] \label{thm:nfdh-empty-volume}
    Let $S$ be a set of hypercubes with side lengths at most $\delta$ and $B$ be a hypercuboidal bin. The \nfdh{} algorithm either packs all the items of $S$ in $B$, or the total volume left empty inside $B$ is at most $\delta\,\surf{B} / 2$.
\end{theorem}

Although these efficiency guarantees of \nfdh{} are restricted to hypercubes, the algorithm can still be useful for packing other geometric forms, such as hyperspheres, in which we are particularly interested in further sections. 
In order to use the \nfdh{} algorithm for other shapes, we wrap the objects in hypercubes.
Given a geometric object $C$, we denote by $\squarehull{C}$ the smallest hypercube, with the same bin orientation, that entirely contains $C$.
We extend this notation for sets, i.e., if $S$ is a set of geometric objects, then $\squarehull{S} = \set{\squarehull{x} : x \in S}$.
Further on, we use this strategy to pack small circles, 
since the wasted area by encapsulating the circles into squares is small.
Next, we state a bound on the height of a bin used to pack a set small squares.

\begin{restatable}{corollary}{thmNfdhSpecificSmallItems}
\label{thm:nfdh-specific-small-items}
    Given $\eps \leq 1/4$ and $w, h, \alpha, \beta \in \Q_+$, with $w \leq h$,
    let $\cali$ be a set of squares whose side lengths are bounded from above by $\bar{s} = \beta \eps^2 w$ and such that $\area{\cali} \leq \alpha \eps w h$.
    There exists $h'>0$ such that 
    the \nfdh{} algorithm packs $\cali$ in a bin of size $w \times h'$ with
    \begin{equation*}
        h' \leq \frac{64\alpha + 16\beta - \beta^2}{64 - 4\beta} \eps h.
    \end{equation*}
\end{restatable}

\section{The Circle Multiple Knapsack Problem}
\label{sec:ckp-main}

Formally, an instance of the \emph{circle multiple knapsack problem} (\cmkp{}) is defined as a tuple $(\cali, w, h, p, m)$ where ${w \leq h \in \Q_+}$ are, respectively, the width and the height of the knapsacks, ${\cali = \srange{s_1,}{s_n}}$ is a set of $n$ circles, each circle $s_i \in \cali$ with diameter $d_i \in \Q_+$ and $d_i \leq w$, ${p \from \cali \to \Q_+}$ is a function of profit on the circles, and $m \in \Z_+$ is the number of available knapsacks. 
We denote by $p_i$ the profit of circle $s_i$.
If~$A$ is a set of circles, we say its profit is $\profits{A} = \sum_{s_i \in A} \profit{i}$.
The objective of the \cmkp{} is to find a packing of a subset $I \subseteq \cali$ of circles in at most $m$ knapsacks of size $w \times h$, maximizing $\profits{I}$.
We denote the optimal value of \cmkp{} for an instance $(\cali, w, h, p, m)$ by $\optMKP{\cali}{m}{w}{h}$.
The \emph{circle knapsack problem} (\ckp{}) is the particular case of \cmkp{} with $m = 1$.
We denote an instance of the \ckp{} by the tuple $(\cali, w, h, p)$ and its optimal value by $\optKP{\cali}{w}{h}$.

In this section, we use the \cmkp{} as an example to illustrate our framework, which results in an \eras{} for the \cmkp{} and a \ptas{} for the \ckp{}.
Hereafter, we define $\reveps = 1/\eps$ and without loss of generality we assume that $\eps \leq 1/4$ and that $\reveps$ and $h\reveps/w$ are integers.

\subsection{Structural Theorem for the \ckp{}}
\label{sec:ckp-structure}

Making use of some properties of a structured packing due to \citet{MiyazawaEtal2015a}, we show that there exists a super-optimal structured packing in an augmented knapsack.
Consider the gap-structured partition procedure as in \cref{sec:circle-packing}.
We show that increasing the height of the knapsack by a small factor allows us to transform a given packing in a structured packing.

\begin{restatable}{theorem}{thmStructuredKnapsack}
\label{thm:structured-knapsack} 
    Let $(\cali,w,h,p)$ be an instance of the \ckp{}.
    For any subset $I \subseteq \cali$ that fits in the knapsack, there is a structured packing of $I$ in an augmented knapsack of size $w \times (1 + 192\eps)h$.
\end{restatable}
\begin{proof}
\newcommand{\medium}[1]{H_{t_{#1}}^{#1}}

    Let $I$ be the set of circles packed in a feasible solution of the \ckp{} instance.
    Using \cref{thm:structured-bin-packing}, we obtain a bound on the total area required to guarantee the existence of a structured packing of $I$, except for the set $H$ of medium items with regard to $I$; then, taking smaller bins, we recursively apply the same procedure over $H$, until all circles are considered.

    Let $\bar{h}_0 = h$ and $\bar{h}_j = 3 \eps \bar{h}_{j-1} = (3 \eps)^j h$, for $j \geq 1$.
    We set $\medium{0} = I$, and for $j \geq 1$, let $\medium{j}$ be a set of the medium items from a gap-structured partition of $\medium{j-1}$, with regard to a grid of size $w \times \bar{h}_{j-1}$, such that $\area{\medium{j}} \leq 2 \eps \area{ \medium{j-1} }$ and $t_j \geq 1$.
    We show that for any $j \geq 0$, there exists a structured packing of $\medium{j} \setminus \medium{j+1}$ into a set $D_j$ of bins that respect $w \times \bar{h}_j$ such that $\area{D_j} \leq (1 + 44\eps) w \bar{h}_j$.

    Since $\medium{0} = I$, the result follows directly from \cref{thm:structured-bin-packing} for $j = 0$.
    We show by induction that $\area{\medium{j}} \leq 2 \eps w \bar{h}_{j-1}$, for $j\geq1$.
    When $j = 1$ we have that $\area{\medium{1}} \leq 2 \eps \area{\medium{0}} \leq 2 \eps w \bar{h}_0$.
    Now assuming that it holds for $j-1$, for $j$ we obtain that
    \begin{equation*}
        \area{\medium{j}} \leq 2 \eps \area{\medium{j-1}} \leq 2\eps (2\eps w \bar{h}_{j-2}) \leq 2\eps w (3 \eps \bar{h}_{j-2}) = 2 \eps w \bar{h}_{j-1}.
    \end{equation*}
    In addition, since $t_j \geq 1$, we know that the diameter of the circles in $\medium{j}$ are bounded from above by $\eps^2 w$.    
    Hence, since when wrapping a circle in a square its area increases by $4/\pi$,
    we have that $\area{\squarehull{\medium{j}{}}} \leq 
    \leq \frac{4}{\pi}\area(H_t^j) \leq 
    \frac{8}{\pi} \eps w \bar{h}_{j-1}$, and their side lengths are at most $\eps^2 w$. Then, from \cref{thm:nfdh-specific-small-items} we conclude that \nfdh{} is able to pack $\squarehull{\medium{j}{}}$ in a bin of size $w \times h'$ with
    \begin{equation*}
        h' \leq \frac{ 64 \cdot \frac{8}{\pi} + 16 - 1 }{64 - 4} \eps \bar{h}_{j-1} \leq 3 \eps \bar{h}_{j-1} = \bar{h}_j.
    \end{equation*}

    Thus, $\medium{j}$ fits in a bin of size $w \times \bar{h}_j$, which implies that
    $\optBP{\medium{j}}{w}{\bar{h}_j} = 1$.
    Therefore, from \cref{thm:structured-bin-packing} we conclude that there is a structured packing of $\medium{j} \setminus \medium{j+1}$ into bins $D_j$ that respect $w \times \bar{h}_j$ and such that $\area{D_j} \leq (1 + 44\eps) \optBP{\medium{j}}{w}{\bar{h}_j} w \bar{h}_j = (1 + 44\eps) w \bar{h}_j$.

    Note that since $\area{\medium{j}} < \area{\medium{j-1}}$, at most $n$ iterations are required to consider all circles of~$I$. Therefore, there exists a structured packing of all circles of $I$ into the set of bins $\cald := \bigcup_{j=0}^{n-1} D_j$ that respect $w \times h$, and whose area is given by
    \begin{align*}
        \area{\cald} &= \sum_{j=0}^{n-1} \area{D_j} \\
            &\leq (1 + 44\eps) w \sum_{j=0}^{n-1} \bar{h}_j \\
            &\leq (1 + 44\eps) w h \sum_{j=0}^{\infty} (3 \eps)^j \\
            &= (1 + 44\eps) w h \paren*{ \frac{1}{1 - 3\eps} }. \\
    \end{align*}
    From the assumption that $\eps \leq 1/4$, we have that $1/(1 - 3\eps) \leq 4$, and therefore
    \begin{align*}
        \area{\cald} &\leq (1 + 44\eps) wh (1 + 12\eps) \\
            &= (1 + 56\eps + 528\eps^2) wh \\
            &\leq (1 + 188\eps) wh,
    \end{align*}
    for $\eps \leq 1/4$.

    At last, given the bound on the area of $\cald$, it remains to obtain a bound on the height of the augmented knapsack that is sufficient to accommodate a structured packing of $\cald$.
    For that, we can simply use \nfdh{}, since the packing of $\cald$ obtained by the \nfdh{} algorithm naturally results in a structured packing.
    Denoting by $\cald'$ the bins of $\cald$ of level $1$ onward, we have that $\area{\cald'} \leq 188\eps w h$ and letting $t$ be the smallest $t_j$ from the sets $\medium{j}$ that originated the sets $D_j$, we have that $t \geq 1$ and thus the side length of the bins of $\cald'$ is bounded from above by $\eps^{2t + 1} w \leq \eps^3 w \leq \frac{1}{4} \eps^2 w$.
    From \cref{thm:nfdh-specific-small-items} we have that \nfdh{} packs $\cald'$ in a bin of size $w \times \Hat{h}$ with
    \begin{equation*}
        \Hat{h} \leq \frac{ 64 \cdot 188 + 16/4 - 1/16 }{ 64 - 1 } \eps h \leq 192 \eps h.
    \end{equation*}

    Hence, we conclude that a knapsack of size $w \times (1 + 192\eps) h$ is sufficiently big to accommodate a structured packing of $\cald$, and consequently of $I$.
\end{proof}

For our algorithm, we use the gap-structured partitioning procedure presented in \cref{sec:circle-packing}, but employing a different scaling on the size of the items and subbins.
We partition $\cali$ in groups $G_i = \set{ s_j \in \cali : \eps^{\reveps i} w \geq d_j > \eps^{\reveps(i+1)} w }$, for $i \geq 0$.
The next steps remain the same: The partitioning of the groups $G_i$ in sets $H_{\ell} = \bigcup_{i \equiv \ell \pmod{\reveps}}{G_i}$; the selection of a set of medium items $H_t$; and the regrouping of $\cali \setminus H_t$ into levels $S_j = \bigcup_{i = t + (j-1)\reveps + 1}^{t + j\reveps - 1} G_i$.
As for the size of the subbins of each level, we set $w_0 = w$, $h_0 = h$ and $w_j = h_j = \eps^{\reveps(t + (j-1)\reveps) + \reveps - 1} w$, for $j \geq 1$.

We choose the medium items as follows.
Let $I^* \subseteq \cali$ be the set of circles of an optimal solution.
For some $1 \leq t < \reveps$, there must be a set $H_t$ such that $\area{H_t \cap I^*} \leq \frac{1}{(\reveps - 1)} w h \leq 2 \eps w h$.
Thus, we fix such index $t$ and handle the medium items $H_t$ separately, by packing a high-profit subset of $H_t$ in a strip of small height.
Then, we exploit the properties of the gap-structured partition to obtain a good packing of the \nonmedium{} items, $\struct{\cali}{t}$.
We note that, despite the change in the scaling, the result of \cref{thm:structured-knapsack} remains valid, since the change only makes the size of the subbins smaller.

\subsection{Packing of the Medium Items}
\label{sec:ckp-packing-medium}

Note that by our choice of $H_t$, we have no information about the profit arising from $H_t$ in some optimal solution $I^*$, i.e., $\profits{H_t \cap I^*}$.
On the other hand, we know that the area of the medium items in an optimal solution is small, more specifically, $\area{H_t \cap I^*} \leq 2\eps w h$, and the medium items themselves are small, that is, their diameter is at most $\eps^\reveps w$, from the fact that $t \geq 1$.
We use these facts to obtain a packing of a subset of $H_t$ in a strip of small height ($\bigO(\eps) h$), and with profit at least $\profits{H_t \cap I^*}$.
This is accomplished by \cref{alg:packing-medium-items}, shown next.

\begin{algorithm}
    \DontPrintSemicolon
    \KwIn{Set $H_t$ of medium items originated from a gap-structured partition, and constant $\eps$.
    }
    \KwOut{Packing of a subset of $H_t$ into a bin of size $w \times 8 \eps h$.}
    Sort $H_t$ in non-increasing order of $p_i / d_i$ \;
    Let $B'$ be a bin of size $(1 + \eps) w \times 4 \eps h$ \;
    Let $j$ be the largest integer for which \nfdh{} is able to pack the first~$j$ items of $\squarehull{H_t}$ in~$B'$ \;
    $P$ := packing of the first $j$ items of $\squarehull{H_t}$ in $B'$ by \nfdh{} \;
    Transform $P$ into a packing in a bin of size $w \times 8 \eps h$ \;
    \Return{$P$} \;    
    \caption{\algPackingMedium{}}
    \label{alg:packing-medium-items}
\end{algorithm}

We will show that \cref{alg:packing-medium-items} actually obtains a high-profit packing of $H_t$.
For that, let us denote by $H_t^* = H_t \cap I^*$ the circles of the medium items that are present in an optimal solution.
First, by making use of the \nfdh{} algorithm and \cref{thm:nfdh-height}, we can obtain a bound on the height that is necessary to pack $H_t^*$.

\begin{restatable}{lemma}{thmMediumItemsBinHeight}
\label{thm:medium-items-bin-height}
    The circles of $H_t^*$ fit in a bin of size $w \times 3 \eps h$.
\end{restatable}

Now, in knowledge of this bound, we can show that the area occupied by the packing obtained in \cref{alg:packing-medium-items} is at least the one of $\squarehull{{H_t^*}}$, which guarantees a high profit due to the ordering of $H_t$.

\begin{restatable}{theorem}{thmPackingMediumItems}
\label{thm:packing-medium-items}
    \cref{alg:packing-medium-items} obtains a packing of a subset of $H_t$ whose profit is at least $\profits{H_t^*}$ in a bin of size $w \times 8 \eps h$.
\end{restatable}

\subsection{Structured Packing of the \Nonmedium{} Items}
\label{sec:ckp-packing-level}

From \cref{thm:structured-knapsack}, we know that there is a super-optimal structured packing for the instance $(\cali, w, h, p)$, if we allow some increase in the size of the knapsack.
Thus, we define ${\Hat{h} = (1 + 192\eps)h}$ to acknowledge such increase in the height of the knapsack, hence in this subsection we are considering the instance $(\struct{\cali}{t}, w, \Hat{h}, p)$.
Given a gap-structured partition $H_t, S_0, S_1, \dots$ of $\cali$, we design an algorithm to find a super-optimal structured packing of only $\struct{\cali}{t}$.

Let $\Hat{\calt}_j = \srange{t_1}{t_{\Hat{T}_j}}$ be the set of different radii among circles of $S_j$, where $\Hat{T}_j = |\Hat{\calt}_j|$, for $j \geq 0$. Each set $\Hat{\calt}_j$ is associated with a tuple $(\range{\Hat{n}_j^1}{\Hat{n}_j^{\Hat{T}_j}})$ of \emph{demands}, where $\Hat{n}_j^k$ is the number of circles of radius~$t_j^k$ in~$S_j$, for $k \in [\Hat{T}_j]$.
A \emph{configuration} of $S_j$ is a tuple $C = (c_1,\ldots,c_{\Hat{T}_j})$ where each $c_k$ is the number of circles of radius $t_j^k$ in $C$, for $k \in [\Hat{T}_j]$. We define $|C| = \sum_{k=1}^{\Hat{T}_j} c_k$ and we say $C$ has $|C|$ circles. The area of a configuration~$C$, denoted by $\area{C}$, is the sum of the area of every circle in~$C$.
We say a configuration $C$ of $S_j$ is \emph{feasible} if its circles fit in one bin of level~$j$.
We denote the set of all feasible configurations of $S_j$ by $\Hat{\calc}_j$.

Due to the way that we define the size of the circles and of the subbins in a same level, we can establish bounds on both the number of circles that can fit into a subbin and the number of feasible configurations, as outlined below.

\begin{restatable}{lemma}{thmNumberConfigurationsPolynomial}
\label{thm:number-configurations-polynomial}
    For any level $j \geq 0$ and configuration $C \in \Hat{\calc}_j$, if $h/w \in \bigO(1)$, then $|C|$ is bounded by a constant $\constantsizeconfig{j}$ and $|\Hat{\calc}_j|$ is bounded by a polynomial in $n$.
\end{restatable}

Hereafter, we refer to a feasible configuration simply as a configuration. We want to determine a subset of configurations (of all levels) that together lead to an optimal structured packing of $\struct{\cali}{t}$, the level items. 
To this end, for a configuration $C \in \Hat{\calc}_{j}$, let $\Hat{f}_{j}(C)$ be the number of empty subbins of size $w_{j+1} \times \Hat{h}_{j+1}$ available for circles of levels $j+1$ onward in a packing of the circles of $C$ in a bin of level $j$.
We present an integer program, named $\Fexact$, to find an optimal structured packing of $\struct{\cali}{t}$ into a knapsack of size $w \times \Hat{h}$.
Consider the following decision variables and formulation.

\begin{itemize}
\item \eqmakebox[items][l]{$x_j^C$: } the number of times configuration $C \in \Hat{\calc}_j$ is used in level $j$;
\item \eqmakebox[items][l]{$b_j$: } the number of empty bins of size $w_j \times \Hat{h}_j$ available for circles of level~$j$;
\item \eqmakebox[items][l]{$z_i$: } binary variable that indicates if circle $s_i \in \struct{\cali}{t}$ is packed or not.
\end{itemize}
\begin{subequations}
\begin{alignat}{4}
(\Fexact) \quad & \omit\rlap{$\max \quad \displaystyle \sum_{s_i \in \struct{\cali}{t}} \profit{i} z_i$} \label{eq:fexact_obj} \\
 & \mbox{s.t.} && \quad &  \sum_{C \in \Hat{\calc}_j} c_k x_j^C &\leq \Hat{n}_j^k     & \qquad & \forall\, j \geq 0, k \in [\Hat{T}_j], \label{eq:fexact_demands} \\
 &             &&       &  \smashoperator[r]{\sum_{s_i \in S_j : r_i = t_j^k}} z_i &= \sum_{C \in \Hat{\calc}_j} c_k x_j^C          &        & \forall\, j \geq 0, k \in [\Hat{T}_j], \label{eq:fexact_z=x} \\
 &             &&       &  \sum_{C \in \Hat{\calc}_j} x_j^C &= b_j              &        & \forall\, j \geq 0, \label{eq:fexact_bins} \\
 &             &&       &  \smashoperator[r]{\sum_{C \in \Hat{\calc}_{j-1}}} \Hat{f}_{j-1}(C) x_{j-1}^C &\geq b_j                   &        & \forall\, j \geq 1, \label{eq:fexact_free-area}\\
 &             &&       &  b_0 &= 1,                                            &        & \label{eq:fexact_b0} \\
 &             &&       &  z_i &\in \binary                                     &        & \forall\, s_i \in \struct{\cali}{t}, \label{eq:fexact_z} \\
&             &&       &  b_j &\in \Z_+                                         &        & \forall\, j \geq 0, \label{eq:fexact_b} \\
&             &&       &  x_j^C &\in \Z_+                                       &        & \forall\, j \geq 0, C \in \Hat{\calc}_j. \label{eq:fexact_x}
\end{alignat}
\end{subequations}

Constraints~\eqref{eq:fexact_demands} ensure that, for every chosen configuration, the demand of each size is not surpassed.
Constraints~\eqref{eq:fexact_z=x} determine which circles are packed, based on the chosen configurations.
Note that the objective function enforces that among circles of equal size, the ones of highest profit are selected. 
Constraints~\eqref{eq:fexact_bins} define the number of subbins used in each level, and constraints~\eqref{eq:fexact_free-area} limit the number of empty subbins available for the subsequent levels, based on the chosen configurations.
Finally, constraint~\eqref{eq:fexact_b0} guarantees that only one knapsack is used and
constraints~\eqref{eq:fexact_z}-\eqref{eq:fexact_x} define the scope of the variables.
 
Note that the number of variables and constraints of $\Fexact$ is bounded by a polynomial in $n$, therefore it is possible to solve its linear relaxation in polynomial time. 
The idea of our algorithm is to obtain an optimal fractional solution to the linear relaxation of $\Fexact$ and then round it up, obtaining a super-optimal solution in an augmented knapsack.
However, directly applying this strategy is not feasible: two key issues arise when considering the current formulation.
First, despite the fact that $\Fexact$ has polynomial size, the bounds presented in \cref{thm:number-configurations-polynomial} are not sufficient to check the feasibility of a configuration in polynomial time via current known algorithms, and thus we are unable to actually build $\Fexact$ in the first place.
Second, even if we could possess $\Fexact$, a fractional solution of its linear relaxation may have too many fractional variables, which could prevent our rounding strategy to yield a solution that causes only a small increase in the knapsack.
To circumvent these issues, we make some modifications to the instance and consider a similar integer program, as described next.

First, we will handle the problem of computing the feasible configurations in polynomial time.
For that, we modify the original instance by rounding the radii of the circles so that we have a constant number of different radii in each level.
Let $R_j = \set{ \rmin{j}(1 + \eps)^k : k \geq 0, \rmin{j}(1 + \eps)^k < \rmax{j} } \cup \set{\rmax{j}} $.
For each level~$j$, we round up the radius of each circle of~$S_j$ to the closest value in~$R_j$. 
We denote the rounded radius of a circle~$s_i$ by~$\Bar{r}_i$, and we refer to them as \textit{scaled circles}.
We define the set of different radii $\calt_j = \srange{t_j^1}{t_j^{T_j}}$ and the demands $(\range{n_j^1}{n_j^{T_j}})$ for the scaled circles analogously as shown previously for the original ones.
With this modification, we now have that the number of different radii in each level is constant.

\begin{restatable}{lemma}{thmBoundDifferentRadii}
\label{thm:bound-different-radii}
    For any level~$j$, the number $T_j$ of different rounded radii is at most~$\reveps^3 \ln(\reveps)$.
\end{restatable}

Now that we have a constant number of different radii, we are able to further bound the number of configurations in each level to be constant, in contrast to the polynomial number in \cref{thm:number-configurations-polynomial} regarding the original instance.

\begin{restatable}{lemma}{thmBoundNumberConfigurations}
\label{thm:bound-number-configurations-scaled}
     For any level~$j$, the number of different configurations of scaled circles of $S_j$ is bounded by a constant.
\end{restatable}

At last, with the new bounds from \cref{thm:bound-different-radii,thm:bound-number-configurations-scaled}, we can use the algorithm from \cref{thm:fkm-etal_bin-packing-in-augmented-bins} to check, in constant time, the feasibility of a configuration of scaled circles and find its corresponding packing in an augmented bin.

\begin{restatable}{lemma}{thmCheckConfigurationFeasibility}
\label{thm:check-configuration-feasibility}
    For any level~$j$, given a configuration $C$ of scaled circles of $S_j$, we can decide if $C$ is feasible, and in the affirmative case, for any constant $\gamma > 0$, we obtain a packing of $C$ in a bin of size $w_j \times (1+\gamma)h_j$, in constant time.
\end{restatable}

Since for each level the number of configurations is constant from \cref{thm:bound-number-configurations-scaled}, with \cref{thm:check-configuration-feasibility} we can compute the sets $\calc_j$ of all feasible configurations of scaled circles of $S_j$, for all the levels $j\geq0$, in polynomial time.
We now define another IP called $\Frounded$, similar to $\Fexact$, to find an optimal structured packing of $\struct{\cali}{t}$ after the scaling. 
The model $\Frounded$ differs from $\Fexact$ in three points.
First, $\Frounded$ considers the scaled circles, and consequently its constraints regard the scaled radii $\Bar{r}_i$, the sets $\calt_j$ and $\calc_j$ and the tuples $n_j$, rather than $r_i$, $\Hat{\calt}_j$, $\Hat{\calc}_j$ and $\Hat{n}_j$.
Second, note that scaling the circles comes with a drawback: the packing of a set of original circles, once scaled, may require additional space to remain feasible.
To compensate this possible increase, in $\Frounded$ we consider augmented bins of size $w'_j \times h'_j$, where $w'_j = (1+\eps)w_j$ and $h'_j=(1+\eps)(1+16\eps)\Hat{h}_j$.
Third, recall that in $\Fexact$, the value $\Hat{f}_{j}(C)$ is the number of empty subbins of size $w_{j+1} \times \Hat{h}_{j+1}$ available for circles of levels $j+1$ onward in a packing of the circles of $C$ in a bin of level $j$.
In $\Frounded$ we estimate such number based on \cref{thm:wasted-area} by defining
\begin{equation*}
    f_{j}(C) =  \frac{w'_{j} h'_{j} - (1+16\eps)\area{C}}{w'_{j+1} h'_{j+1}},
\end{equation*}
which is a lower bound on the number of empty subbins of size $w'_{j+1} \times h'_{j+1}$ after packing a configuration $C \in \calc_j$ in a bin of size $w'_j \times h'_j$.
The model $\Frounded$ is then described in the following.
The decision variables $x$, $z$ and $b$ have the same meaning as in $\Fexact$.
\begin{subequations}
\begin{alignat}{4} 
 (\Frounded) \quad & \omit\rlap{$\max \quad \displaystyle \smashoperator[lr]{ \sum_{s_i \in \struct{\cali}{t}} } \profit{i} z_i$} \label{eq:frounded_objective} \\
 & \mbox{s.t.} && \quad &  \sum_{C \in \calc_j} c_k x_j^C &\leq n_j^k     & \qquad & \forall\, j \geq 0, k \in [T_j], \label{eq:frounded_demands} \\
 &             &&       &  \smashoperator[r]{\sum_{s_i \in S_j : \Bar{r}_i = t_j^k}} z_i &= \sum_{C \in \calc_j} c_k x_j^C          &        & \forall\, j \geq 0, k \in [T_j], \label{eq:frounded_z=x} \\
 &             &&       &  \sum_{C \in \calc_j} x_j^C &= b_j              &        & \forall\, j \geq 0, \label{eq:frounded_bins} \\
  &             && \quad & \smashoperator[r]{\sum_{C \in \calc_{j}}} f_{j}(C) x_{j}^C &\geq b_{j+1}        & \qquad & \forall\, j \geq 0, \label{eq:frounded_free-area} \\
 &             &&       &  b_0 &= 1,                                                    &        & \label{eq:frounded_b0} \\
 &             &&       &  z_i &\in \binary                                             &        & \forall\, s_i \in \struct{\cali}{t}, \label{eq:frounded_z} \\
 &             &&       &  b_j &\in \Z_+                                                &        & \forall\, j \geq 0, \label{eq:frounded_b} \\
 &             &&       &  x_j^C &\in \Z_+                                              &        & \forall\, j \geq 0, C \in \calc_j. \label{eq:frounded_x}
\end{alignat}
\end{subequations}

Although scaling the circles and using a lower bound for the exact value of $\Hat{f}_j$ reduce the available space for achieving a feasible packing, the use of an augmented knapsack with dimensions $w' \times h'$ compensates for these approximations.
Consequently, the optimal value provided by $\Frounded$ is guaranteed to be at least as large as that of $\Fexact$.

\begin{restatable}{lemma}{thmFroundedFexact}
\label{thm:frounded>=fexact}
    Given an instance $(\struct{\cali}{t}, w, \Hat{h}, p)$, we have $\opt(\Frounded) \geq \opt(\Fexact)$.
\end{restatable}

Thus, with these modifications on the instance and on the first formulation, we are able to compute all feasible configurations and obtain an optimal solution to the linear relaxation of $\Frounded{}$ whose value is super-optimal, all in polynomial time.
It remains to obtain an integer solution from the fractional one.
However, we are still faced with the problem of not having control in how the fractional variables are distributed among the levels, and hence, simply obtaining an arbitrary optimal fractional solution does not suffice.

Note that each variable $x_j^C$ indicates how many subbins of level $j$ are used to pack circles using configuration $C$, thus the $x$ variables dictate the total space used in the packing.
Rounding up a variable $x_j^C$ comes with the expense of increasing the area of the packing by the size of a subbin of level $j$, in the worst case.
Thus, if a fractional solution contains a large number of fractional $x$ variables in the first levels, rounding the $x$ variables up would imply in large additional space.
To guarantee that the extra space is small, we need that the fractional solution must present the following two properties.
One, there cannot be more than one fractional $x$ variable in level $0$, since any of such variables when rounded up would correspond to another knapsack, whose area is prohibitive; 
two, the number of fractional $x$ variables in levels~$1$ onward must be small enough so that the extra bins due to the rounding up fit in a strip of small height.
We will denote a fractional solution satisfying these properties as \textit{balanced}.
We now show how to obtain a balanced solution.

The occurrence of fractional $x$ variables in level~$0$ can be prevented by simply not relaxing their integrality in the linear relaxation of $\Frounded$.
Let $\milpFrounded$ be the MILP obtained by relaxing the integrality of all variables of $\Frounded$, except for $x_0$.
Since the number of configurations in each level is constant (\cref{thm:bound-number-configurations-scaled}), the model $\milpFrounded$ has a constant number of integer variables, therefore it can be solved in polynomial time (a result of \citet{Lenstra1983}).

Now to ensure a small number of fractional variables at each other level, we take advantage of a neat characteristic of our model:
the formulation exhibits a nearly decomposable block structure.
Let us call by \textit{block} $j$ of $\Frounded$ the set of constraints \eqref{eq:frounded_demands}--\eqref{eq:frounded_free-area} corresponding to level $j$, along with their associated decision variables.
Notably, block $j$ is almost exclusively composed of the variables $x$, $z$, and $b$ associated with level $j$. The only exception is the constraint \eqref{eq:frounded_free-area}, which involves the variable $b_{j+1}$.
If $b_{j+1}$ could be excluded from the variables in block $j$, the blocks of $\Frounded$ would become disjoint, allowing them to be solved independently and thereby fully decomposing $\Frounded$.

It turns out that this decomposition is indeed achievable.
Let $\Frounded[b^*]$ be the model where the variables $b$ are fixed to specific values $b^*$, by replacing the occurrences of $b_j$ and $b_{j+1}$ in constraints~\eqref{eq:frounded_bins} and~\eqref{eq:frounded_free-area} by $b^*_j$ and $b^*_{j+1}$, respectively, and modifying constraints~\eqref{eq:frounded_free-area} to enforce equality.
We then call by $\Flevel[j][b^*]$ the block $j$ of $\Frounded[b^*]$.
To decompose $\Frounded$, we first compute an optimal fractional solution $(x^*, z^*, b^*)$ for $\milpFrounded$ and then fix the values of the variables $b$ with $b^*$, solving $\Flevel[j][b^*]$ for each level $j \geq 1$.
Thus, the model $\Frounded$ is effectively solved twice: first as a whole, and then on a block-by-block fashion, as shows \cref{alg:balanced-fractional-solution} next.

\begin{algorithm}
\DontPrintSemicolon
\KwIn{Instance $(\struct{\cali}{t}, w, \Hat{h}, p)$}
\KwOut{Balanced solution to $\milpFrounded$}
$(x^*, b^*, z^*) \gets \text{optimal solution of } \milpFrounded$ \;
\For{each $j \geq 1$} {
    $(\Tilde{x}_j, \Tilde{z}_j) \gets \text{linear relaxation of } \Flevel[j][b^*]$ \; \label{alg:balanced-fractional-solution-line-flevel}
}
$\Tilde{x} \gets (x_0, \Tilde{x}_1, \Tilde{x}_2, \ldots)$ \;
$\Tilde{z} \gets (z_0, \Tilde{z}_1, \Tilde{z}_2, \ldots)$ \;
\Return{$(\Tilde{x}, b^*, \Tilde{z})$} \;
\caption{\algBalanced{}}
\label{alg:balanced-fractional-solution}
\end{algorithm}

Decomposing $\Frounded$ this way provides a key advantage: in each block $j$, the number of constraints -- and consequently the number of fractional variables -- is bounded by $\bigO(T_j)$.
This ensures the attainment of a balanced solution, as desired.

\begin{restatable}{lemma}{thmBoundNumberRoundedUpVariables}
\label{thm:bound-number-rounded-up-variables}
    \cref{alg:balanced-fractional-solution} returns an optimal balanced solution of $\milpFrounded$ for instance $(\struct{\cali}{t}, w, \Hat{h}, p)$ with at most $2T_j + 2$ non-null variables in each level $j\geq1$, in time $\bigO({ 2^{r^{3r^6}} \poly(n, r^{r^6})})$.
\end{restatable}

Hereafter, given an $n$-dimensional vector $x = (x_1,\ldots,x_n)$, we define the \emph{ceil} of $x$ as $\ceil{x} = (\ceil{x_1},\ldots,\ceil{x_n})$.
Let $(x^*, b^*, z^*)$ be an optimal balanced fractional solution of $\milpFrounded$ given by \cref{alg:balanced-fractional-solution}. 
We round the variables $x^*$ up to the next integer, yielding a collection of configurations represented by the vector $\ceil{x^*}$.
Because the solution is balanced, the total additional area required to account for the extra bins introduced by the rounding is small.
\begin{restatable}{lemma}{thmStripForExtraBins}
\label{thm:strip-for-extra-bins}
    Let $(x^*, b^*, z^*)$ be an optimal balanced fractional solution of $\milpFrounded$. 
    The extra bins created after rounding the variables $x^*$ to $\ceil{x^*}$ fit into a strip of size $w' \times \eps h'$.
\end{restatable}

After rounding the $x$ variables of a balanced fractional solution, we have a set of configurations from all levels that describes a super-optimal structured packing of the scaled \nonmedium{} items in an augmented knapsack.
To actually build a packing out of this set of configurations, we first obtain a packing of each configuration in a bin of its respective level, then we place these packings within the knapsack by recursively drawing grids of the size of each level and using the entirely empty grid cells.
\begin{restatable}{lemma}{thmConfigurationToPacking}
\label{thm:configuration-to-packing}
    Given an instance $(\struct{\cali}{t}, w, \Hat{h}, p)$,
    for each level~$j$, let $X_j$ be a collection (allowing duplication) of configurations of the scaled circles of $S_j$, with respect to bins of size $w'_j \times h'_j$.
    Given a constant $\gamma > 0$, there is a polynomial-time algorithm that finds a packing of maximum profit of the original circles that correspond to the configurations of $X_j$ in bins of size $w'_j \times (1 + \gamma) h'_j$.
\end{restatable}

\begin{algorithm}
\DontPrintSemicolon 
\KwIn{Instance $\calz = (\struct{\cali}{t}, w, \Hat{h}, p)$ and constant $\eps$}
\KwOut{A super-optimal solution to $\calz$ in a knapsack of size $(1 + \bigO(\eps)) w \times (1 + \bigO(\eps)) h$
}

\For{each level $j \geq 0$} {
    Let $R_j = \set{ r_{\min}^j(1 + \eps)^k : k \geq 0, r_{\min}^j(1 + \eps)^k < r_{\max}^j } \cup \set{r_{\max}^j} $ \;
    Scale the circles of $S_j$ by rounding up their radii to values of $R_j$ \;
    Obtain the set of feasible configurations $\calc_j$ of the scaled circles of $S_j$ \;
}
$(x^*, b^*, z^*) \gets \algBalanced{}(\calz)$ \; \label{alg:structured-packing-line-algbalanced}
Build a packing $P$ from the configurations of $\ceil{x^*}$ into a knapsack of size $w \times (1+\bigO(\eps))h$ \; \label{alg:structured-packing-line-build-packing-from-configurations}
\Return{packing $P$}
\caption{\algStructuredPacking}
\label{alg:structured-packing}
\end{algorithm}

Finally, joining all these steps we give an algorithm that, given an instance $(\struct{\cali}{t},w,\Hat{h},p)$ coming from a gap-structured partition, it computes a super-optimal packing of the level items $\struct{\cali}{t}$ in an augmented knapsack.
See \cref{alg:structured-packing}.

\begin{restatable}{theorem}{thmPtasCkpWhConstant}
\label{thm:packing-remaining-items}
    Given an instance $(\struct{\cali}{t}, w, \Hat{h}, p)$ 
    of \ckp{} with $h/w \in \bigO(1)$ and a constant $\eps \leq 1/4$, \cref{alg:structured-packing} obtains a packing of a subset $I \subseteq \struct{\cali}{t}$ of circles in a knapsack of size $(1 + \eps)w \times (1+1911\eps)h$ such that 
    $\profits{I} \geq \optKP{\struct{\cali}{t}}{w}{h}$, in polynomial time in the size of the instance.
\end{restatable}
\begin{proof}
    Let $(x^*, b^*, z^*)$ be the optimal balanced fractional solution obtained in \cref{alg:structured-packing-line-algbalanced} of \cref{alg:structured-packing}. 
    After rounding up the solution, the profit can only increase.
    From \cref{thm:strip-for-extra-bins}, we know that the configurations corresponding to $\ceil{x^*}$ fit into a knapsack of size $w' \times (1+\eps)h'$. 
    We then use \cref{thm:configuration-to-packing} with $\gamma = \eps$ to find a packing~$P'$ of circles corresponding to the configurations given by $\ceil{x^*}$ into bins of size $w'_j \times (1 + \eps)^2 h'_j$, for each level $j$.
    Thus $P'$ fits in a knapsack of size $(1 + \eps) w \times (1 + \eps)^2 h'$.
    Since $(1 + \eps)^2 h' = (1 + \eps)^3 (1 + 16\eps) (1 + 192\eps) h \leq (1 + 1911\eps) h$ for $\eps \leq 1/4$, we have that $P'$ fits into a knapsack of size $(1 + \eps) w \times (1 + 1911\eps) h$.
    Since $h/w \in \bigO(1)$, the number of circles in a configuration is bounded by a constant and thus the algorithm from \cref{thm:check-configuration-feasibility}, that we use to check the feasibility of a configuration and obtain its corresponding packing, runs in constant time.
    Therefore, \cref{alg:structured-packing} runs in polynomial time when $h/w \in \bigO(1)$.
\end{proof}

\subsection{Efficient Resource Augmentation Scheme for \ckp{} and \cmkp{}}
\label{sec:ckp-complete-algorithm}

Now, by combining the procedures presented in this section, we present our framework applied for the \ckp{} problem, designed to achieve an efficient resource augmentation scheme.
It consists of deriving a gap-structured partition of the instance (\cref{sec:ckp-structure}) and, after guessing a set of medium items, it finds a super-optimal packing of the medium items (\cref{sec:ckp-packing-medium}) and a super-optimal packing of the \nonmedium{} items in an augmented knapsack (\cref{sec:ckp-packing-level}). 
It then combines the two packings together to form a super-optimal packing of the whole instance in an augmented knapsack. 
Among all guesses of medium items, we return the resulting packing of maximum profit. See \cref{alg:ckp-resource-augmentation-scheme}.

\begin{algorithm}
\DontPrintSemicolon
\KwIn{Instance $(\cali, p, w, h)$ of \ckp{} and constant $\eps$.}
\KwOut{A super-optimal packing in a knapsack of size $(1 + \bigO(\eps)) w \times (1 + \bigO(\eps)) h$.}
Let $\reveps = 1/\eps$.\;
Define $G_i = \set{ s_j \in \cali : \eps^{\reveps i}w \geq d_j > \eps^{\reveps(i+1)}w }$, for $i \geq 0$.\;
Define $H_\ell = \set{ G_i : i \equiv \ell \pmod{\reveps} }$, for $0 \leq \ell < \reveps$.\;
\For{each $t$ from $1$ to $\reveps-1$} { \label{alg:ckp-resource-augmentation-scheme-line-t}
    Define $S_j = \bigcup_{i = t+(j-1)\reveps+1}^{t+j\reveps-1} G_i$, for every integer $j \geq 0$.\;
    Define $w_0 = w$, $h_0 = h$, and $w_j = h_j = \eps^{\reveps (t+(j-1)\reveps) + \reveps - 1}w$, for $j \geq 1$.\;
    $P_H \gets \algPackingMedium{}((H_t, w, h, p), \eps)$. \; 
    $P_\cals \gets \algStructuredPacking{}((\struct{\cali}{t}, w, \Hat{h}, p), \eps)$.\; 
    $P_t \gets$ $P_H$ stacked on top of $P_\cals$.\;
}
\Return{packing $P_t$ of maximum profit.}
\caption{\algScheme{}}
\label{alg:ckp-resource-augmentation-scheme}
\end{algorithm}

\begin{theorem}
\label{thm:alg-ckp-wh-constant}
    Given an instance $(\cali, w, h, p)$ of \ckp{} with $h/w \in \bigO(1)$ and a constant $\eps \leq 1/4$, \cref{alg:ckp-resource-augmentation-scheme} obtains a packing of a subset $I \subseteq \cali$ of circles in a knapsack of size $(1 + \eps)w \times (1+1919\eps)h$ such that 
    $\profits{I} \geq \optKP{\cali}{w}{h}$, in polynomial time in the size of the instance.
    in time $\bigO(n^c f(1/\eps))$, where $c$ is a constant independent of $\eps$ and $f$ is a computable function.
\end{theorem}
\begin{proof}
    Let $I^* \subseteq \cali$ be the set of circles of an optimal solution. 
    Note that in some iteration of \cref{alg:ckp-resource-augmentation-scheme-line-t} the set $H_t$ will be such that $\area{H_t \cap I^*} \leq 2 \eps wh$. 
    Thus, in such iteration $t$, line $7$ obtains a super-optimal packing of the medium items in a knapsack of size $w \times 8\eps h$, from \cref{thm:packing-medium-items}, and line $8$ obtains a super-optimal packing of the \nonmedium{} items in a knapsack of size $w \times (1 + 1911\eps) h$, from \cref{thm:packing-remaining-items}. 
    Thus, $P_t$ consists of a super-optimal packing in a knapsack of size $(1 + \eps) w \times (1 + 1919\eps) h$.
    
    Regarding time complexity, the partitioning of the instance takes linear time,
    \cref{alg:packing-medium-items} (\algPackingMedium{}) has a running time of $\bigO(n \log^2 n)$ by using binary search,
    and the running time of \cref{alg:structured-packing} (\algStructuredPacking{}) is
    of type $\bigO(n^c f(1/\eps))$ due to \cref{thm:check-configuration-feasibility,thm:bound-number-rounded-up-variables}.
\end{proof}

Additionally, with few modifications \cref{alg:ckp-resource-augmentation-scheme} works for the \cmkp{}, that is, when multiple knapsacks are available rather than just one.

\begin{theorem}
\label{thm:alg-cmkp-wh-constant}
    Let $(\cali, w, h, p, m)$ be an instance of \cmkp{}. If $h/w \in \bigO(1)$, then for any constant $\eps > 0$ we can obtain, in polynomial time, a packing of $I \subseteq \cali$ in at most $m$ knapsacks of size $(1 + \eps) w \times (1 + 1920\eps)h$ such that $\profits{I} \geq \optMKP{\cali}{m}{w}{h}$.
\end{theorem}
\begin{proof}

    We show which modifications are necessary in each of the three main steps of \cref{alg:ckp-resource-augmentation-scheme}, namely the gap-structured partitioning, the packing of the medium items and the packing of the \nonmedium{} items.

    Starting with the gap-structured partitioning, let $I^*$ be the circles of an optimal solution.
    Since we now have $m$ knapsacks, we choose a set $H_t$ of medium items such that $\area{H_t \cap I^*} \leq 2 \eps \area{I^*} \leq 2 \eps mwh$, and $t \geq 1$.
    The structural theorem presented in \cref{sec:ckp-structure} works the same, since given any feasible packing for the \cmkp{} instance, we can transform it into a structured packing where each knapsack has size $w \times (1 + 192\eps) h$ by simply applying \cref{thm:structured-knapsack} for each knapsack individually.

    Regarding the packing of the medium items, we can use \cref{alg:packing-medium-items} considering a knapsack of size $w \times mh$. Then by \cref{thm:packing-medium-items}, we have a super-optimal packing of the medium items in a bin $B$ of size $w \times 8\eps mh$.
    We only need to transform $B$ into $m$ bins of small height. Since the packing obtained by \cref{alg:packing-medium-items} follows a shelf-like structure due to the \nfdh{} algorithm, we can simply partition $B$ in $m$ strips of height $8\eps h$. To avoid intersection with packed items, it suffices to    increase the height of each strip to match the base of the next closest shelf. 
    The size of the medium items is bounded by $\eps^\reveps w$, since $t \geq 1$, therefore this procedure makes each of the $m$ strips have height at most $8 \eps h + \eps^\reveps w \leq 9 \eps h$.

    At last, to obtain a super-optimal packing of the \nonmedium{} items, it suffices to change constraint \eqref{eq:frounded_b0} of $\Frounded$ to $b_0 \leq m$. This modification does not affect the behavior of \cref{alg:structured-packing}, so the result of \cref{thm:packing-remaining-items} remains the same for the \cmkp{} problem.

    Therefore, applying \cref{alg:ckp-resource-augmentation-scheme} with the aforementioned modifications gives us a super-optimal packing for the $(\cali, w, h, p, m)$ instance in at most $m$ knapsacks of size $(1 + \eps) w \times (1 + 1920) h$.
\end{proof}

Recall that \cref{alg:structured-packing}, and consequently \cref{alg:ckp-resource-augmentation-scheme},
runs in polynomial time only under the assumption that $h/w \in \bigO(1)$.
Nevertheless, in possession of the result of \cref{thm:alg-cmkp-wh-constant}, we can eliminate this restriction, thus establishing an \eras{} for the \cmkp{} with unconstrained number of bins $m$ and bin size ratio $h/w$.

\begin{theorem}
\label{thm:alg-ckp-wh-unbounded}
    There is an \eras{} for the \cmkp{}.
\end{theorem}
\begin{proof}
    First we show that we can transform any packing $P$ in a bin of size $w \times h$ into another packing $P'$ in a bin $B'$ of size of $w \times (1 + 4\eps) h$ so that $B'$ can be partitioned in strips of size $w \times w/\eps$ in such a way that no circles in $P'$ overlap the boundaries of the strips.
    For that, we start by partitioning the original bin of size $w \times h$ in strips of size $w \times w/\eps$.
    Since each circle has diameter at most $w$, a region of size $w \times 2w$ centered at each strip contains all circles that overlap the boundaries.
    We can thus move all these circles to a bin of width $w$ and height $h/(w/\eps) \cdot 2w = 2\eps h$.
    
    Now the original bin of size $w \times h$ has the desired property of no circle overlapping the boundary of the strips of size $w \times w/\eps$, but we are left with a bin of size $w \times 2\eps h$ where this may not be true.
    We then apply the same procedure recursively in this bin, until the desired property holds true for all the bins created.
    Stacking these bins on top of each other, we obtain a new bin $B'$ with the desired property and whose height is
    \begin{equation*}
        h' \leq h + \sum_{i=1}^{\infty} (2\eps)^i h = h + \frac{2\eps}{1 - 2\eps} h \leq (1 + 4\eps) h,
    \end{equation*}
    since $1/(1-2\eps) \leq 2$ for $\eps \leq 1/4$.

    Given an instance $(\cali, w, h, p, m)$ of \cmkp{}, we define $q = h'/(w/\eps) = (1 + 4\eps) \eps h/w$ and create the instance $\calz' = (\cali, w, w/\eps, p, qm)$ of \cmkp{}. From the previous result, we know that $\optMKP{\cali}{qm}{w}{w/\eps} \geq \optMKP{\cali}{m}{w}{h}$.
    Since $(w/\eps)/w = \reveps \in \bigO(1)$, we can use the algorithm of \cref{thm:alg-cmkp-wh-constant} to obtain a super-optimal packing for instance $\calz'$ in at most $qm$ bins of size $(1 +\eps) w \times (1 + 1920\eps) w/\eps$.
    Stacking each group of $q$ bins of size $w \times w/\eps$ on top of each other, we obtain $m$ bins of width $(1 + \eps) w$ and height $q (1 + 1920\eps) w/\eps = (1 + 1920\eps) (1 + 4\eps) h \leq (1 + 3844\eps) h$ for $\eps \leq 1/4$.
\end{proof}

\subsection{Resource Augmentation in only One Dimension}
\label{sec:ckp-aug-scheme-only-one-dimension}

The previous results use resource augmentation in both dimensions.
This is due to the scaling of the circles and subbins done in \cref{alg:structured-packing}.
We are able to remove the necessity of resource augmentation in the width of the bin, leaving only the height augmented.
For that, we handle level~$0$ in a particular manner.
The idea is to scale the circles of level $0$ by a more fine-grained factor, so that we can use the shifting algorithm of \cref{thm:shifting-algorithm} to obtain a packing of the scaled circles in a bin with height augmented by $\eps$.

To formally present the modifications, we first consider back the \ckp{} problem under the assumption that $h/w \in \bigO(1)$.
Recall that $\constantsizeconfig{0} = \frac{4}{\pi} \reveps^{2\reveps^2 - 2\reveps + 2} \frac{h}{w}$ is the bound on the number of circles of level $0$ that fit in a bin, as calculated in \cref{thm:number-configurations-polynomial}.
Instead of scaling the circles of level $0$ by powers of $(1 + \eps)$, we define $\delta = \eps^2/(6{\constantsizeconfig{0}}^2)$ and scale the circles by powers of $(1 + \delta)$, namely the values from the set $R_0 = \set{ \rmin{0} (1 + \delta)^k : k \geq 0, \rmin{0} (1 + \delta)^k < \rmax{0} } \cup \set{\rmax{0}}$.
Despite the fact that $\delta$ is much smaller than $\eps$, the number of different scaled sizes remains constant.

\begin{restatable}{lemma}{thmBoundDifferentRadiiLevelZero}
\label{thm:bound_different_radii_level_zero}
    The number $T_0$ of different sizes of scaled circles of level $0$ is bounded by a constant, under the assumption that $h/w \in \bigO(1)$.
\end{restatable}

Moreover, by scaling the circles of level $0$ to powers of $1 + \delta$, any feasible packing is turned into a $\delta h$-packing, and thus using the result of \cref{thm:shifting-algorithm} we can convert it into a packing in a bin of size $w \times (1 + \eps) h$.

\begin{restatable}{lemma}{thmShiftingLevelZero}
\label{thm:shifting-level-zero}
    For any packing $P$ of circles of level $0$ in a bin $B_{w \times h}$, there is another packing $P'$ of the scaled circles in a bin of size $w \times (1 + \eps) h$.
\end{restatable}

Rearranging the circles of level~$0$ this way is already sufficient to obtain a solution that uses resource augmentation in only one dimension, and the time complexity of our framework remains the same.

\begin{theorem}
    Given an instance $(\cali, w, h, p)$ of \ckp{} with $h/w \in \bigO(1)$ and a constant $\eps$, 
    we can obtain in polynomial time a packing of a subset $I \subseteq \cali$ of circles in a knapsack of size $w \times (1+\bigO(\eps))h$ such that $\profits{I} \geq \optKP{\cali}{w}{h}$.
\end{theorem}
\begin{proof}
    We use \cref{alg:ckp-resource-augmentation-scheme} only changing the scaling of the circles of level $0$ as explained here.
    Since $T_0$ is bounded by a constant, the number of configurations of level $0$, which is bounded by $(\constantsizeconfig{0}+1)^{T_0}$, also remains constant, and thus we are still able to solve $\milpFrounded$ in polynomial time.
    Therefore, \cref{alg:ckp-resource-augmentation-scheme} continues to take polynomial time with this change and it gives us a super-optimal packing in a bin of size $(1 + \eps) w \times h'$ where $h' = (1 + \bigO(\eps)) h$.
    We use the result of \cref{thm:shifting-level-zero} to convert the packing of the circles of level $0$ in a bin $B_{(1 + \eps) w \times h}$ to another one in a bin $B'_{w \times (1 + \eps) h'}$.
    When doing this conversion, we also move any subbins of the further levels to the empty space of this new packing in $B'$, as necessary.
    Note that since $\area{B} = \area{B'}$ and the configuration $C$ of circles of level $0$ remains the same, the value of $f(C)$ considering the bin $B'$ cannot be lower than its value in $B$, and therefore there is always enough space in the new packing to place the subbins.
    After changing the packing of the circles of level $0$, we have the guarantee that the rightmost strip of size $\eps w \times h'$ contains only subbins of level $1$ onward, and thus we can rearrange them in a strip of size $w \times \eps h'$.
    Stacking this strip on top of $B'$, we obtain a packing of the circles in a bin of size $w \times (1 + \bigO(\eps)) h$.
\end{proof}

At last, in possession of this result, we can apply the same procedures done in \cref{thm:alg-cmkp-wh-constant,thm:alg-ckp-wh-unbounded} to obtain an equivalent theorem for the more general case of the \cmkp{} with unconstrained ratio between $h$ and $w$ and unconstrained number of knapsacks $m$.

\begin{theorem}
    There is an \eras{} using resource augmentation in only $1$ dimension for the \cmkp{}.
\end{theorem}

\subsection{\ptas{} for the \ckp{} and \eaugptas{} for the \cmkp{}}
\label{sec:ckp-ptas}

In this section, we show how the techniques we designed to obtain a resource augmentation scheme for \ckp{} naturally supports the development of a \ptas{} for the same problem.
That is, our framework can be used to either avoid loss of profit by using resource augmentation, or have a small loss of profit while maintaining the original size.
Replacing the use of resource augmentation to having a small loss os profit is very convenient in our framework primarily due to the well-behaved utilization of the augmented resource in our gap-structured partitioning.
Specifically, the circles in the augmented resource are not packed arbitrarily; rather, they are organized into sets of subbins from level $1$ onward.
This structure allows us to focus on moving entire subbins into the original knapsack, rather than dealing with individual circles.
Additionally, since the subbins of a level are much smaller than the circles of the previous level, an appropriate calibration of the range of radii that falls in a same group enables a good use of any empty space between big circles to pack smaller ones, making the use of resource augmentation dispensable. 
This way, some guarantee of empty space between circles of level~$0$ is enough to pack a big amount of bins of level~$1$ onward.

Recall that the use of resource augmentation in the \eras{} comes from three steps of our framework:
$i)$ the packing of the medium items;
$ii)$ the scaling of the circles and bins;
and $iii)$ the accommodation of the extra bins resulting from the rounding of a fractional solution of $\milpFrounded{}$.

The first one is easy to resolve: It suffices to choose the medium items to have negligible profit, instead of small area, and then we simply discard them.
Regarding the second obstacle, the scaled subbins of level $1$ onward are easy to handle: Since the circles are very small, we are able to remove a strip of low profit to turn the sizes of the scaled subbins back to the original ones.
The scaling of circles of level $0$, on the other hand, is not admissible, since scaling them implies in scaling the knapsack itself, leading to an inevitable use of resource augmentation.
However, we know that even when the circles are not scaled, the number of feasible configurations is still bounded by a polynomial from \cref{thm:number-configurations-polynomial}.
Thus, instead of scaling the level $0$, we compute $\Hat{\calc}_0$ in polynomial time and guess which configuration is present in an optimal solution.

That leaves us with the third obstacle, which is the accommodation of the extra bins.
Here we take advantage of the guarantee of empty space among circles of level~$0$. 
A simple observation is that due to the geometry of circles, as long as they have a minimum size, there is always some space within the knapsack that cannot be occupied by any circle, and therefore is always left available for the next levels.
This property aligns very well with our gap-structured partitioning:
The larger the size gap between the circles of level $0$ (big circles) and the subbins of level $1$ (small bins), the greater the unoccupied space left by the big circles within the knapsack, and consequently, the more small bins can fit into this space.
Thus, by knowing that the number of subbins fitting in the original knapsack significantly exceeds those in the augmented region, the fact that our framework mannerly fills the augmented region with subbins turns it easy to simply replace subbins so as to leave the ones of lowest profit outside the original knapsack, allowing us to discard the resource augmentation while losing negligible profit.

The next lemma states the empty area guarantee in a packing of large circles.
To this end, the geometric characteristics of circles make it intuitive to explore the empty space in the corners of the knapsack.

\begin{restatable}{lemma}{thmEmptyAreaBigCircles}
\label{thm:empty-area-big-circles}
    Given a packing of circles with radius at least $\delta$ in a rectangular bin $B$, there is an empty area of at least $\brackets{(1 - 1/\sqrt{2}) \delta}^2$ in $B$.
\end{restatable}

This result is sufficient to obtain a \ptas{} from our resource augmentation scheme.

\begin{theorem} \label{thm:ckp-ptas}
    There is a \ptas{} for the circle knapsack problem.
\end{theorem}
\begin{proof}
    Given an instance $(\cali, w, h, p)$ of the \ckp{}, we obtain a gap-structured partition as in \cref{sec:ckp-structure}, but we consider a set $H_t$ of medium items such that $\profits{H_t \cup I^*} \leq \bigO(\eps) \profits{I^*}$, where $I^*$ is an optimal solution. 
    This way we can discard $H_t$ losing only a factor of $\eps$ of the optimal value.

    We compute the set $\Hat{\calc}_0$ of all feasible configurations of $S_0$ without scaling the circles and we do not scale the knapsack. By \cref{thm:number-configurations-polynomial}, $|\Hat{\calc}_0|$ is polynomially bounded.
    For each $C \in \Hat{\calc}_0$, we set
    \begin{equation*}
        f_{1}(C) =  \frac{w h - (1+16\eps)\area{C}}{w'_1 h'_1},
    \end{equation*}
    and apply our IP technique from \cref{sec:ckp-packing-level} to levels~$1$ onward. 
    Note that, in this case, we solve the linear relaxation of $\Frounded{}$, instead of $\milpFrounded{}$, which certainly can be done in polynomial time.
    With this, we obtain a super-optimal solution of the \nonmedium{} items in an augmented knapsack, as in \cref{sec:ckp-packing-level}. However, since the circles of level~$0$ are not scaled, no circle of $S_0$ intersect the augmented area.
    As a consequence, the resource augmentation comes from two sources: the use of scaled subbins from level $1$ onward; and the extra bins due to the rounding of the fractional balanced solution.
    We show how to dispose of the use of resource augmentation while losing little profit.
    
    To remove the resource augmentation due to the scaling of the subbins, it suffices to handle the bins of level~$1$. 
    Since the diameter of any circle of level~$1$ is at most $\eps w_1$, we can partition each scaled subbin , of size $w'_1 \times w'_1$, into $1/\bigO(\eps)$ strips of width $\bigO(\eps)w'_1$.
    By removing one strip of lowest profit, we lose only a factor of $\eps$ in profit of the circles packed in the scaled subbin, and leave enough space to accommodate all the remaining circles
    within a subbin of size $w_1 \times w'_1$. Repeating the process along the height, we obtain a high-profit packing in a subbin of size $w_1 \times w_1$.

    It remains to handle the extra bins due to the rounding of the fractional solution.
    With the bound on the number of extra bins from \cref{thm:bound-number-rounded-up-variables} along with their small sizes, we can conclude that the extra area they occupy is very small, i.e, a value $A \leq \eps^{2rt + 2r - 5}wh$.
    Also, from \cref{thm:empty-area-big-circles}, we know that there is an empty area $E \geq \eps^{2\reveps t} (1 - 1/\sqrt{2})^2 w^2$ between the packed circles of $S_0$ that is always left available to level~$1$.
    Note that the empty area among circles of level~$0$ is much greater than the area of the extra bins, i.e, $E/A \geq \bigO(r^{2r-5})$.
    We partition the set of all subbins of level~$1$ in the solution into groups of area $A$, then the cheapest group surely has profit at most a factor of $\eps$ of the total profit of the packing.
    We move the subbins of lowest profit to the augmented area, place all the remaining subbins within the knapsack of size $w \times h$, and then remove the augmented area.
    Since in all steps the profit loss is a factor of $\eps$ of the total profit of the packing in the augmented knapsack, the resulting packing in the knapsack of original size has profit at least $(1 - \bigO(\eps))$ of the original packing, which was superoptimal, thus leading to a $\ptas{}$.
\end{proof}

Observe that the only step preventing our algorithm from achieving an \eptas{} lies in the computation of $\Hat{C}_0$ to handle the level items.
Since $\size{\Hat{C}_0}$ is polynomial in $n$, enumerating all configurations requires $\bigO(n^{f(1/\eps)})$ time.
Nevertheless, all the other steps remain efficient.

This procedure naturally extends to the \cmkp{}. If the number of knapsacks, $m$, is constant, we can still guess the $m$ configurations of $\Hat{C}_0$ that belong to an optimal solution in polynomial time, thus achieving a $\ptas$.
Even when $m$ is not constant, an $\eaugptas$ is obtained, as follows. Without loss of generality, we assume that $m \geq r(2T_j + 2)$.
Then we simply solve the linear relaxation of $\Frounded$, while including the continuous $x_0$ variables.
Based on the bound on the number of extra bins, the additional knapsacks created due to rounding up the $x_0$ variables are limited to $2T_j + 2$.
However, by our choice of $m$, this is at most $\eps m$. Consequently, we can discard the $\eps m$ knapsacks with the lowest profit, resulting in an augmented feasible solution that uses no more than $m$ knapsacks while losing only a factor of $\eps$ of the total profit.
Importantly, this approach avoids any enumeration over polynomial-sized sets, therefore preserving the algorithm's efficiency.
This leads to the following.

\begin{theorem}
    There is an $\eaugptas{}$ for the circle multiple knapsack problem, and a $\ptas{}$ when the number of knapsacks is bounded by a constant.
\end{theorem}

\section{Other Packing Problems}
\label{sec:other-packing-problems}

In contrast with the knapsack problem, where the available area is delimited and we must choose a subset of the items to be packed, 
another common class of packing problems consists in minimizing the area used to pack all items.
One example is precisely the bin packing problem (\cref{prob:bp}), and in particular, the \cbp{}. 
Other examples are the multiple strip packing problem (\cref{prob:msp}) and the multiple minimum container problem (\cref{prob:mmc}).
We denote by \cmsp{} and \cmmsb{} \cref{prob:msp} and \cref{prob:mmc}, respectively, restricted to circles as items.
In this section, we show that our framework also yields efficient schemes for these problems, namely an \eras{} for the \cbp{} and an \eptas{} for the \cmmsb{} and the \cmsp{}.

\subsection{\eras{} for the \cbp{}}

Recall that \cref{alg:ckp-resource-augmentation-scheme} has three main steps: the gap-structured partition, the packing of the medium items and the packing of the \nonmedium{} items. 
The first two steps are applied in exactly the same way.
Regarding the \nonmedium{} items, we only need to adapt $\Frounded$ to handle the \cbp{} constraints, rather than those of \ckp{}. 
For that, only two changes are needed. The objective function becomes $\min b_0$, ensuring the minimization of the number of bins used in the packing; and constraints \eqref{eq:frounded_demands}, regarding which items are packed, become equalities. 
After that, \cref{alg:structured-packing} works in the same way, leading to a packing of all the circles in at most $\optBP{\cali}{w}{h}$ bins of size $w \times (1+\bigO(\eps)h)$, i.e., a resource augmentation scheme for the \cbp{}.

We observe that \citet{MiyazawaEtal2015a} already gave a resource augmentation scheme for the \cbp{}, however, our algorithm allows the extra feature of easily dealing with demand on the items, as shown in \cref{sec:problems_with_multiplicity}.

\subsection{\eptas{} for the \cmmsb{} and the \cmsp{}}

In our framework, to obtain a gap-structured partition, the circles are first grouped based on the ratio between their radii and the width of the bin. 
However, in the \cmmsb{} problem, the side length of the square bin is not defined a priori, since that is precisely what must be determined.
The overall idea to handle this is to derive a set of candidate side lengths for the bin.
For each candidate, we compute a gap-structured partition of the circles and compute the feasible configurations of each level. 
Then, we run the IP $\Frounded{}$  over all candidates at once, adding to the model decision variables to ensure that only one side length is used, and changing the objective function to choose the smallest side length.

Hereafter, we consider an instance $(\cali, m)$ of \cmmsb{} and denote by $l^* := \optMMSB{\cali}{m}$ its optimal value.
In order to obtain the candidate side lengths, we first need to determine lower and upper bounds on this value.
For a candidate side length $l$ to be feasible, the area of the bins must be at least the area of the circles, i.e, $m l^2 \geq \area{\cali}$. Thus, $\ubar{l} = \sqrt{\area{\cali}/m}$ is a lower bound on $l^*$.
An upper bound can be computed by encapsulating the circles in their square hulls and packing them via \nfdh{}. Let $m'$ be the number of bins used in such packing. Since in a packing obtained by \nfdh{} all bins, except the last one, are surely filled with a density higher than $1/4$, we have that $\area{\squarehull{\cali}} \geq (m' - 1) l^2 / 4$. For circles, this implies that $\Bar{l} = \sqrt{32/\pi} \sqrt{\area{\cali}/m}$ is an upper bound on $l^*$.
Now we discretize the range $[\ubar{l}, \Bar{l}]$ of candidate side lengths for $l^*$ using powers of $(1 + \eps)$. Formally, we consider the set $\call = \set{ (1 + \eps)^k \cdot \ubar{l} : k \geq 0, (1 + \eps)^k \cdot \ubar{l} < \Bar{l} } \cup \set{ \Bar{l} }$.
We define $L = \size{\call}$ and denote by $\Tilde{l}_i$ the $i$th smallest value in $\call$, for $i \in [L]$. 
Note that $L$ is of the order of $\log_{1 + \eps}(\Bar{l} / \ubar{l}) = \log_{1 + \eps} \sqrt{32/\pi}$, thus constant.

For each $i \in [L]$, we compute a gap-structured partition $H_t^i, S_0^i, S_1^i, \ldots$, with $\area{H_t^i} \leq 2\eps\area{\cali}$ and $t \geq 1$.
We denote $\structsuper{\cali}{t}{i} = \cup_{j \geq 0} S_j^i$, for $i \in [L]$.
Then the framework starts by obtaining a packing of the \nonmedium{} items.
First, with few modifications, we adapt the IP $\Frounded$ to the \cmmsb{}. %
For each $i \in [L]$, we add decision variables $y_i$ that indicate whether the side length $\Tilde{l}_i$ is used, along with constraints to ensure that only one $\Tilde{l}_i$ is in fact chosen.
We name this new model $\Frounded^\mathrm{MMSB}$.

Prior to solving the linear relaxation of $\Frounded^\mathrm{MMSB}$, we scale the circles and compute the feasible configurations, in the same way done in \cref{sec:ckp-packing-level}.
Formally, for each $\Tilde{l}_i \in \call$, we scale the radii of the circles of each level $j\geq 0$ to powers of $1+\eps$. After the scaling, we obtain the set of all feasible configurations $\calc^i_j$, the number $T^i_j$ of different radii and their respective demands $n_{i,j}^k$, for each $k \in [T^i_j]$, for each level $j\geq 0$. Note that \cref{thm:bound-number-configurations-scaled} still holds, i.e., $|\calc^i_j|$ is bounded by a constant, for all $i \in [L]$ and $j\geq 0$.

\begin{subequations}
\begin{alignat}{4}
(\Frounded^\mathrm{MMSB}) \quad & \omit\rlap{$\min \quad \displaystyle \sum_{i \in [L]} \Tilde{l}_i y_i$} \\
 & \mbox{s.t.} &&       & \sum_{C \in \calc_j^i} x_j^C c_k &= n_{i,j}^k y_i       & \qquad & \forall\, i \in [L], j \geq 0, k \in [T_j], \label{eq:fmmsb_demands} \\
 &             &&       & \sum_{C \in \calc_j^i} x_j^C &= b_j^i                   &        & \forall\, i \in [L], j \geq 0, \label{eq:fmmsb_bins} \\
 &             &&       & \smashoperator[r]{\sum_{C \in \calc_{j-1}^i}} f_{j-1}(C) x_{j-1}^C &\geq b_j^i    &        & \forall\, i \in [L], j \geq 1, \label{eq:fmmsb_free-area} \\
 &             &&       &  b_0^i &= m y_i                                         &        & \forall\, i \in [L], \label{eq:fmmsb_b0} \\
 &             &&       &  \sum_{i \in [L]} y_i &= 1,                                        &        & \label{eq:fmmsb_sum-y} \\
 &             &&       &  y_i &\in \binary                                       &        & \forall\, i \in [L], \\
 &             &&       &  b_j^i &\in \Z_+                                        &        & \forall\, j \geq 0, \\
 &             &&       &  x_j^C &\in \Z_+                                        &        & \forall\, i \in [L], j \geq 0, C \in \calc_j^i. \label{eq:fmmsb_x}
\end{alignat}
\end{subequations}

Constraints \eqref{eq:fmmsb_demands}-\eqref{eq:fmmsb_b0} have the same meaning as in $\Frounded{}$.
Constraint~\eqref{eq:fmmsb_sum-y} guarantees that only one among all candidate side lengths is chosen, while constraints~\eqref{eq:fmmsb_b0} ensure that only configurations related to the selected side length are used in the solution. At last, the objective function guarantees the use of the smallest side length.

Since both the number $L$ of candidate side lengths and the number of configurations in each level 
are bounded by constants, we can solve, in polynomial time, the linear relaxation of $\Frounded^\mathrm{MMSB}$ maintaining the integrality of the $y$ variables and the $x_0^C$ variables, for all configurations of level~$0$, for every side length $\Tilde{l}_i \in \call$.
Let $\Tilde{l}^*_i$ be the optimal value of such relaxation.
We have that $\Tilde{l}^*_i \leq \optMMSB{\cali \setminus H^i_t}{m} \leq \optMMSB{\cali}{m}$.
By applying the same algorithm of \cref{sec:ckp-packing-level} we obtain a packing of the level items in bins of size $\Hat{l} = (1 + \bigO(\eps)) \Tilde{l}^*_i \leq (1+\bigO(\eps)) \optMMSB{\cali}{m}$.

It remains to pack the medium items.
From the above result, we have that $\area{\cali \setminus H^i_t} \leq m \Hat{l}^2$, which implies that $\area{\cali} \leq \frac{1}{1 - 2\eps} m \Hat{l}^2$.
Thus, $\area{H^i_t} \leq 2\eps \area{\cali} \leq \frac{2\eps}{1 - 2\eps} m \Hat{l}^2 \leq 4\eps m \Hat{l}^2$ for $\eps \leq 1/4$.
Then we use the \nfdh{} algorithm in the same manner as in \cref{thm:medium-items-bin-height} to pack $H^i_t$ in a strip of size $\Hat{l} \times \bigO(\eps) m \Hat{l}$.
We also apply the same procedure used for the medium items in \cref{thm:alg-cmkp-wh-constant} to transform such packing into a packing in $m$ strips of size $\Hat{l} \times \bigO(\eps) \Hat{l}$. 
Finally, by merging the packing of the \nonmedium{} items with that of the medium items, we obtain a packing of $\cali$ in $m$ square bins of side length $(1 + \bigO(\eps)) \Hat{l} \leq (1 + \bigO(\eps)) \optMMSB{\cali}{m}$, yielding an \eptas{}.

\begin{theorem}
    There is an \eptas{} for the \cmmsb{} problem.
\end{theorem}

Note that the only aspect where \cmsp{} differs from \cmmsb{} is that the width of the strips is fixed.
In this case, we consider a set of candidate heights, instead of side lengths, in the same manner done for \cmmsb{}.
Additionally, we employ the same strategy used in \cref{sec:ckp-aug-scheme-only-one-dimension}
to restrict the use of resource augmentation to only the height of the bin.
Combining these procedures, we obtain an \eptas{} for the \cmsp{}.

\begin{theorem}
    There is an \eptas{} for the \cmsp{} problem.
\end{theorem}

\section{Generalization to Fat Objects}
\label{sec:generalization-fat-objects}

We define a fat object by the concept of two-ball fatness, as follows.
For an object $o$, we call a largest hypersphere enclosed by $o$ an \textit{\insc{} sphere} and denote its diameter by $\fatd{o}$. Similarly, we call a smallest hypersphere enclosing $o$ a \textit{\circums{} sphere} and denote its diameter by $\fatD{o}$.
In case there are multiple \insc{} and \circums{} spheres for $o$, we fix one of them arbitrarily without loss of generality.
The \textit{slimness factor} of $o$ is given by the ratio $\fatD{o} / \fatd{o}$. Then we say an object $o$ is \textit{fat} if its slimness factor is bounded by a constant.
Specifically, given a constant $\psi \geq 1$, we say that $o$ is a \emph{\fat{$\psi$} object} if $\fatD{o}/\fatd{o} \leq \psi$. 

By this definition, hyperspheres represent the fat objects with the smallest possible slimness factor, i.e., they are \fat{$1$} objects.
Thus, when considering only circles in \cref{sec:ckp-main}, we presented our algorithm for the particular case of fat objects with slimness factor $\psi = 1$ and dimension $d = 2$.
Now we show that our algorithm is not restricted by these bounds and can handle any convex $d$-dimensional \fat{$\psi$} object, for constants $\psi$ and $d$. 
For simplicity, we hereafter consider fixed constants $\psi$ and $d$, and we denominate a convex $d$-dimensional $\psi$-fat object simply as object.

For our framework, we need to ensure the following:
\begin{enumerate}
    \item there is a polynomial-time algorithm to decide if a set of (a constant number of) objects can be packed in a given bin, and if so, it returns a packing of the objects in an augmented bin (possibly in all dimensions);

    \item only a small amount of volume is wasted due to the (discarded) subbins of the subsequent level that partially intersects an object.
\end{enumerate}

Our technique by itself guarantees the second requirement, since it allows us to calibrate the granularity of the gap-structured partition so that the subbins of a level are sufficiently much smaller than the items of the previous level. This way, we ensure that the total volume of the bins of a level $j$ that partially intersect objects of level $j-1$ is sufficiently small.

As for the first requirement, to obtain such algorithm, we use the same algebraic apparatus employed for circles. This approach makes our algorithm work for a wide class of convex fat objects. 
Essentially, it suffices that the objects can be described by a system of polynomial equalities and inequalities. To ensure polynomial time, the size of the system must be constant, i.e., it must have a constant number of equalities and inequalities and variables. 
Nevertheless, in case this is not true, we can approximate the object to a new one meeting this requirement. 
As a consequence, 
we do not need to impose many restrictions on the shapes of the input objects, opposed to most of other works, which are restricted to convex polytopes, for instance.
The only condition is that the number of different shapes is constant.

We also remark that our framework supports a wide range of objects for the bins as well. Akin to the items, it must be a convex object and it must be described by a system of polynomial inequalities. We observe that the shape of the bin affects the computation of a packing only in the first level, because from level~$1$ onward, the objects are packed in hypercubes again due to the use of the grid partition strategy.

\subsection{Resource Augmentation Scheme for Fat Objects}
\label{sec:ras-for-fat-objects}

We denote by \fokp{$d, \psi, \tau$} the knapsack problem for convex $d$-dimensional \fat{$\psi$} objects containing at most $\tau$ different shapes.
An instance of \fokp{$d, \psi, \tau$} is given by a tuple $(\cali, l, \profit{})$, where $\cali$ is an input set of objects, $l = (l_1,\ldots,l_d)$ is the size of a hypercuboidal knapsack $B$, and $p$ is a function of profit on the objects.
Without loss of generality, we assume $l_{\min} = l_1 \leq \ldots \leq l_d = l_{\max}$.

Recall that, in short, our framework presented in \cref{sec:ckp-main} consists of three main steps: obtaining a gap-structured partition, packing the medium items, and packing the level items. 
Now we present the necessary adaptations to handle the \fokp{$d, \psi, \tau$} problem, with constant $d$, $\Psi$ and $\tau$.

\subsubsection{Gap-structured partition}

To obtain a gap-structured partition, we classify the items based on the diameter of the objects' \circums{} spheres and the smallest side of the bin, and we keep the ratio of $\eps^\reveps$ between groups.
Formally, we define groups $G_i = \set{ o \in \cali : \eps^{\reveps{}i} l_{\min} \geq \fatD{o} > \eps^{\reveps{}(i+1)} l_{\min} }$. 
In the same manner as before, we partition these groups into sets $H_{\ell} = \bigcup_{i \equiv \ell \pmod{\reveps}}{G_i}$, for $0 \leq \ell < \reveps$ and, given an index $1 \leq t < \reveps$, we define the \nonmedium{} objects $\struct{\cali}{t}$ as the sets $S_j = \bigcup_{i = t + (j-1)\reveps + 1}^{t + j\reveps - 1} G_i$, for $j \geq 0$. 
The medium items are chosen similarly: being $I^* \subseteq \cali$ the set of objects of an optimal solution, we choose a set $H_t$ such that $\vol{H_t \cap I^*} \leq 2 \eps \vol{B}$.
We denote by $\fatDmin{j}$ and $\fatDmax{j}$ the diameters of the \circums{} spheres of the smallest and largest objects in~$S_j$, respectively.

\subsubsection{Packing \nonmedium{} items}

Again, we pack the \nonmedium{} objects by levels, using bins of appropriate size in each level. Here, to define the size of the subbins, we use $l_{\min}$ as reference. Specifically, the objects of $S_0$ are packed in $B$, and for $j \geq 1$, we set $w_j = \eps^{\reveps(t + (j-1)\reveps) + \reveps - 1} l_{\min}$, representing the side length of the hypercubes (subbins) used to pack the items from level~$1$ onward, in their respective levels.
Then we employ the same process used in \cref{sec:ckp-packing-level}: scale the objects, obtain a balanced solution of the MILP and build a (super-optimal) packing in an augmented knapsack.

To obtain a balanced fractional solution 
we need to bound the number of different objects formats in each level.
To this end, we again apply a rounding procedure over the original objects. 
For each level $j \geq 0$ we define $R_j = \set{ \fatDmin{j}(1 + \eps)^k : {k \geq 0}, \fatDmin{j}(1 + \eps)^k < \fatDmax{j} } \cup \set{\fatDmax{j}}$.
For each level $j$, fixing one shape, we round up the diameter of the \circums{} sphere of each object of that shape to the closest value of $R_j$, and we scale the object accordingly.
We refer to the objects after the rounding as \emph{scaled objects} and we refer to each different scaled object as a \emph{pattern}.
The calculations from \cref{thm:bound-different-radii} remains exactly the same, which implies that the number of different patterns from the same shape remains $T_j \leq \reveps^3 \ln{\reveps}$.
We apply the same rounding procedure for each of the $\tau$ shapes. 
Thus, the number of different patterns in total is bounded by $\tau \reveps^3 \ln{\reveps}$.
Therefore, similarly as in \cref{thm:number-configurations-polynomial}, we can obtain a good bound on the number of configurations in each level.

\begin{restatable}{lemma}{thmFatObjectsBoundNumberConfigurations}
\label{thm:fat-objects-bound-number-configurations-scaled}
     For any level~$j$, the number of different configurations of scaled objects of $S_j$ is bounded by a constant.
\end{restatable}

With these bounds on the number of configurations, we use the IP $\Frounded{}$ in the same way as in \cref{sec:ckp-packing-level}. 
Then in level~$0$ there is no extra bins, and in each level~$j\geq 1$, the number of extra bins is limited by the number of constraints of the IP $\Flevel$, which in turn is bounded by $\bigO(\tau T_j)$ (the number of different patterns in level~$j$).
We show that the extra bins due to the rounding of a fractional solution fit in a small $d$-dimensional strip. 
Recall that we denote by $\milpFrounded{}$ the MILP obtained by relaxing the integrality of all variables of $\Frounded{}$, except for $x_0$.

\begin{restatable}{lemma}{thmFatObjectsStripForExtraBins}
\label{thm:fat-objects-strip-for-extra-bins}
    Let $(x^*, b^*, z^*)$ be an optimal balanced fractional solution of $\ \milpFrounded{}$ for an instance $(\cali, l, \profit{})$ of the \fokp{$d, \psi, \tau$}. 
    The extra bins created after rounding the variables $x^*$ to $\ceil{x^*}$ fit into a $d$-dimensional strip of size $(l_1, \ldots, l_{d-1}, \eps l_d)$.
\end{restatable}

At last, we mention that the feasibility of a configuration is checked using a generalization of the algebraic apparatus presented in \cref{sec:preliminaries} to $d$ dimensions. More details can be seen in \cite{MiyazawaEtal2015a}.

This concludes the packing of the \nonmedium{} objects.

\subsubsection{Packing the medium items}

To pack the medium objects, the idea is the same explained for circles.
We enclose the objects in hypercubes and pack the maximum number of such hypercubes ordered by relative profit using the \nfdh{} algorithm.

Recall that for a geometric object $o$, we denote by $\squarehull{o}$ the smallest hypercube that entirely contains $o$.
Let $\squarehull{V_o}$ be the volume of $\squarehull{o}$.
Intuitively, since the objects are fat, the volume of the enclosed hypercube is also bounded by a constant factor of the volume of the object.
This is shown in next lemma.

\begin{restatable}{lemma}{thmRatioVolumeFatObjectHypercube}
\label{thm:ratio-volume-fat-object-hypercube}
    For any $d$-dimensional \fat{$\psi$} object $o$, it holds that $\squarehull{V_o} / V_o \leq 
    \paren*{\sqrt{2} \psi}^d$.
\end{restatable}

Now, we use the same idea of \cref{alg:packing-medium-items} to obtain a high-profit packing of $H_t$. 
With the result of \cref{thm:ratio-volume-fat-object-hypercube} combined with the guarantees of used volume given by \nfdh{} (\cref{sec:preliminaries}), we know that \nfdh{} can fill a volume at least as big as that used by the medium objects of an optimal solution. Due to the ordering of the objects by relative value in \cref{alg:packing-medium-items}, the profit of the packed objects is at least that of the medium objects of an optimal solution, as desired.

This settles the medium items. Then, we can enunciate the final result. 

\begin{theorem}
    There is an \eras{} for the \fokp{$d, \psi, \tau$} problem for any constants $d, \psi, \tau$.
\end{theorem}

Furthermore, we observe that when the objects have a \emph{lifting property}, the resource augmentation can be restricted to only one dimension.
Essentially, this property consists in being possible to rearrange the objects within the bin, by slights shifts in a way that only the height of the bin is increased.
For the objects that do not present such property, 
we can use some algorithm based on discretization to pack the objects into one bin; however, the resource augmentation is in all dimensions. 
One algorithm of this nature was presented by Bansal et al.~\cite{BansalEtal_APTAS-d-dimensional-bin-packing_2006} for packing hypercubes into the minimum number of unit bins. 
At last, we claim that although a discretization approach could also be applied for objects with the lifting property, the algebraic approach gives a better time complexity.

In addition, the results presented in \cref{sec:ckp-ptas} and \cref{sec:other-packing-problems} also extends to the $d$-dimensional case.
First, regarding hyperspheres, note that the empty area guarantee of \cref{thm:empty-area-big-circles} on the $2$-dimensional context extends straightforwardly to $d$ dimensions. 

\begin{restatable}{corollary}{thmEmptyVolumeBigHyperspheres}
\label{thm:empty-volume-big-hyperspheres}
    Given a packing of $d$-dimensional hyperspheres with radius at least $\delta$ in a hypercuboidal bin $B$, there is an empty volume of at least $\brackets{(1 - 1/\sqrt{d}) \delta}^d$ in $B$.
\end{restatable}

Then the same strategy applied in \cref{sec:ckp-ptas} yields an equivalent result for hyperspheres.

\begin{theorem}
    There is an \eaugptas{} for the hypersphere multiple knapsack problem, and a \ptas{} when the number of knapsacks is bounded by a constant.
\end{theorem}

Finally, the results regarding the problems in \cref{sec:other-packing-problems} extends to convex fat objects as follows.

\begin{theorem}
    There is an \eptas{} for the MMCP of convex fat objects and for the MSP of convex fat objects with a lifting property.
\end{theorem}

\section{Additional Features}
\label{sec:additional-features}

In this section we present the introduction of additional features that can be considered over the problems addressed previously.
Namely, our framework supports the introduction of multiplicity on the items, rotation on the items, and also additional constraints in the knapsack problems.

\subsection{Allowing Rotation}

We now show that our framework handles rotation on the items, i.e., in a feasible packing, each object can be rotated by arbitrary angles.

Given that, in a $d$-dimensional space, the number of rotational degrees of freedom of a rigid object is $\binom{d}{2}$, we denote a \emph{rotation} as a tuple $\alpha = (\alpha_{11\mathstrut}, \alpha_{12\mathstrut}, \ldots, \alpha_{1d\mathstrut}, \alpha_{22\mathstrut}, \alpha_{23\mathstrut}, \ldots, \alpha_{2d\mathstrut}, \alpha_{33\mathstrut}, \ldots, \alpha_{dd\mathstrut})$, where each $\alpha_{ij}$ is an angle associated with the $ij$-plane.
Then, to \emph{rotate} an object $o$ by a rotation $\alpha = (\alpha_{11\mathstrut}, \ldots, \alpha_{dd\mathstrut})$ means rotating $o$ in each $ij$-plane by the angle $\alpha_{ij}$, for $1 \leq i < j \leq d$. 
We say that a rotation $\alpha$ \emph{respects} an angle $\theta$ if every $\alpha_{ij}$ is a multiple of $\theta$.
We extend the term and say that a packing \emph{respects} an angle $\theta$ if every packed object assumes a rotation that respects $\theta$.
We define the center of the \insc{} sphere of an object as the point of reference for any rotation. We suppose that rotations are always clockwise and we consider the $\arccos$ function restricted to the interval $[0,\pi]$.

First we show that, given a set $\cali$ of objects and a packing $P$ of $\cali$ in a bin $B$, there exists another packing $P'$ of $\cali$ that respect a certain angle $\alpha$ in an augmented bin $B'$, where $\alpha$ is bounded by a constant and $B'$ is augmented in all dimensions by a small constant factor.
Knowing that this holds, our algorithm then consider only packings that respect $\alpha$, and due to the bounds on $\alpha$, the time complexity to obtain the feasible packings does not increase asymptotically.
To show the existence of this modified packing, the first step is to show that in any feasible packing $P$, we can scatter the objects in a strategic way so that we create some extra space between any two objects, without introducing overlaps.
This brings some freedom to rotate each object by some angle. 
The following lemma regards the scattering procedure.

\begin{lemma}\label{thm:sparsed-packing}

    Consider a value $\gamma > 0$ and a set $\cali$ of convex $d$-dimensional \fat{$\psi$} objects, with $\fatd{o} \geq \gamma$ for every $o\in\cali$, that fit in a $d$-dimensional bin $B$ of size $(\range{l_1}{l_d})$.
    For any $\delta > 0$, there is a packing of $\cali$ where the distance between 
    any object and any other object or the boundary of the bin is at least $\delta$, in an augmented bin of size $(\range{l'_1}{l'_d})$ with $l'_i \leq \paren*{ 1 + \frac{d \delta}{\gamma} } l_i + \delta(1 + \sqrt{d})$ for each dimension $i$.
\end{lemma}
\begin{proof}
    Let $P$ be a packing of $\cali$ in $B$.
    Let $\delta' = \delta \sqrt{d}$ and consider a dimension~$i$.
    We partition the bin $B$ into strips of length $\Delta = \gamma/\sqrt{d}$ in dimension $i$.
    Then we assign each object of $\cali$ to the strip in which the center of its \insc{} sphere belongs to.
    Denoting by $O_k$ the objects that were assigned to the $k$th strip, we shift each object of $O_k$ in the increasing direction of dimension $i$ by $k \delta'$.
    By doing so, the objects of $O_k$ are shifted by $\delta'$ in relation to the ones of $O_{k-1}$, for $k > 1$, and the objects of $O_1$ are shifted by $\delta'$ in relation to the left boundary of the bin.
    We iteratively do this shifting in every dimension $i \in [d]$, obtaining a new packing $P'$.

    We now show that $P'$ fulfills the desired properties.
    Regarding the distance between an object and the left boundary of the bin, it suffices to note that objects assigned to the first strip are shifted by $\delta' \geq \delta$.
    Now we consider two objects $A$ and $B$ of $\cali$.
    There must be at least one dimension in which $A$ and $B$ are each assigned to a different strip after the shifting procedure.
    If not, the centers of the \insc{} spheres of $A$ and $B$ would both be inside a hypercube of length $\Delta < \gamma$, contradicting the feasibility of the packing $P$.
    Moreover, by convexity, there exists a hyperplane $H$ separating $A$ and $B$ in $P$.
    When we apply the shifting procedure over $P$, the hyperplane $H$ is shifted as well. 
    Let $H'$ be such equivalent hyperplane after the shifting.
    The distance between $A$ and $B$ is at least the distance between $H$ and $H'$, thus it suffices to measure the distance between these hyperplanes.
    
    \begin{figure}[htb!]
        \centering
        \begin{center}

\begin{tikzpicture}[
    dot/.style={draw, circle, fill=black, minimum size=0.01cm}
]
\begin{scope}


    \coordinate (Hstart) at (-2, 2);
    \coordinate (Hend) at ( 2,-2);
    \coordinate (Hshift) at (0.75, 0.75);

    \draw (Hstart) -- (Hend) node[pos=0, left] {$H$};
    \draw[dashed] ($ (Hstart) + (Hshift) $) -- ($ (Hend) + (Hshift) $) node[pos=0, left] {$H'$};


    \def \stripHeight {2.7};
    \draw[dotted] (-1.2, \stripHeight) -- (-1.2, -\stripHeight);
    \draw[dotted] ( 0.8, \stripHeight) -- ( 0.8, -\stripHeight);

    \coordinate (stripHintersection) at (-1.2, 1.2);
    \coordinate (stripbottom) at (-1.2, -\stripHeight);

    \pic [draw, angle radius=13] {angle = stripbottom--stripHintersection--Hend};
    \pic [draw, angle radius=11] {angle = stripbottom--stripHintersection--Hend};


    \coordinate (Acenter) at (-1,-1);
    \coordinate (Bcenter) at ( 1, 1);
    \coordinate (triangVertex) at (Bcenter |- Acenter);
    
    \def \radius {1.41};

    \fill[] (Acenter) circle[radius=2pt];
    \draw[] (Acenter) circle (\radius);
    \node[] at (-1.6*\radius, -1.6*\radius) {$A$};
    
    \fill[] (Bcenter) circle[radius=2pt];
    \draw[] (Bcenter) circle (\radius);
    \node[] at (1.6*\radius, 1.6*\radius) {$B$};


    \draw (Acenter) -- (Bcenter) node[pos=0.15,above=0.15] {$d_A/2$} node[pos=0.6,above=0.15] {$d_B/2$};
    \draw (Acenter) -- (triangVertex) node[pos=0.5,below] {$d_j(A, B)$};
    \draw (Bcenter) -- (triangVertex);

    \pic ["$\theta_j$", draw=black, angle eccentricity=1.5, angle radius=15] {angle = triangVertex--Acenter--Bcenter};
    \pic [draw, angle radius=13] {angle = triangVertex--Acenter--Bcenter};
    
    \pic [draw, angle radius=6] {right angle = Acenter--triangVertex--Bcenter};


    \coordinate (vStart) at (1.2,-1.2);
    \coordinate (vEnd) at ($ (vStart) + (Hshift) $);
    \coordinate (vjEnd) at ($ (vStart) + (1.47,0) $);

    \path (vStart) edge[-latex] node[pos=0.4, above] {$\va{v}'$} (vEnd);
    \path (vStart) edge[-latex] node[pos=0.5, below] {$\va{u}_j$} (vjEnd);

    \pic [draw, angle radius=4] {right angle = vEnd--vStart--Hstart};

    \pic [draw, angle radius=9] {angle = vjEnd--vStart--vEnd};
    \pic [draw, angle radius=7] {angle = vjEnd--vStart--vEnd};
    
\end{scope}
\end{tikzpicture}

\end{center}
        \caption{
            Schematic of objects $A$ and $B$ projected in the plane $Q_j$, considering the worst case scenario in which the \insc{} spheres of $A$ and $B$ touch. In the figure, $A$ and $B$ are shown before the shifting, with $H$ being their separating hyperplane, and $H'$ is the hyperplane that separates them after the shifting. The vertical dotted lines are the hyperplanes used to define the strips. The vector $\va{v}'$ represents $\va{v}$ projected in the plane $Q_j$, and $\va{u}_j$ represents the shifting of $B$ relative to $A$ in dimension $j$.
        }
        \label{fig:sparsed-packing-distance-calculation}
    \end{figure}
    
    For that purpose, let $\va{v}$ be the vector orthogonal to $H$ that represents the shifting of $H$ to $H'$.
    For each dimension $i$, we denote by $\va{v}_i$ its component in this dimension and $\theta_i$ the angle between $\va{v}$ and $\va{v}_i$.
    From the fact that $\sum_{i \in [d]} \cos^2\paren*{\theta_i} = 1$, there must be one $j$ such that $\cos\paren*{\theta_j} \geq 1/\sqrt{d}$.
    We now consider such dimension $j$.
    Let $Q_j$ be the plane defined by vectors $\va{v}$ and $\va{v}_j$ and consider the projection of the packing $P'$ in $Q_j$, as shows \cref{fig:sparsed-packing-distance-calculation}.
    Note that the angle between the hyperplane $H$ and the hyperplanes used to define the strips in dimension $j$ is also $\theta_j$.
    Thus, denoting by $d_j(A, B)$ the distance in dimension $j$ between the centers of the \insc{} spheres of $A$ and $B$, we have that $d_j(A, B) \geq (\fatd{A} + \fatd{B}) \cos\paren*{\theta_j}/2 \geq \gamma \cos\paren*{\theta_j} \geq \gamma / \sqrt{d} = \Delta$.
    This implies that objects $A$ and $B$ were assigned to different strips during the shifting procedure in dimension $j$, 
    and consequently $B$ was shifted by at least $\delta'$ more than $A$ in a direction of dimension $j$.
    Thus, letting $\va{u}_j$ be the vector in this direction of dimension $j$ that represents the shifting of $B$ relative to $A$ in such dimension, we have that $\| \va{u}_j \| \geq \delta'$, and by projecting $\va{u}_j$ on $\va{v}$, we conclude that
    \begin{equation*}
        \| \va{v} \|
        \geq \| \va{u}_j \| \cos\paren*{\theta_j}
        \geq \delta' \frac{1}{\sqrt{d}}
        = \delta,
   \end{equation*}
    guaranteeing the desired distance between any two objects in $P'$.

    Lastly, regarding the size of the augmented bin, note that the number of strips used to partition $B$ in dimension $i$ is $\ceil{l_i / \Delta}$.
    Thus, the objects that were assigned to the last strip were shifted by $\ceil{l_i / \Delta} \cdot \delta'$, and by adding an extra $\delta$ to the length of the augmented bin's side to guarantee that the distance between such objects and the right boundary of the bin is at least $\delta$, we have that a sufficient length for the augmented bin is
    $l'_i = l_i + \ceil{l_i / \Delta} \cdot \delta' + \delta \leq \paren*{ 1 + \frac{d \delta}{\gamma} } l_i + \delta(1 + \sqrt{d})$.
\end{proof}

As the second step, we show that given some empty space around an object, it can be rotated by at least some amount without traversing such empty space.
For an object $o$, we denote by $\proj{ij}{o}$ the projection of $o$ in the $ij$-plane.
For some value $\lambda > 0$, we define the \emph{$\border{\lambda}{}$} of $o$ as the set of points that are at a distance of exactly $\lambda$ of $o$.
We simply write $\border{\lambda}{ij}$ of $o$ to refer to the $\border{\lambda}{}$ of $\proj{ij}{o}$.
The next lemma bounds the value of the angle by which an object can be rotated in each $ij$-plane without traversing its $\border{\lambda}{}$.

\begin{lemma}\label{thm:fat-object-rotation}
    Given a value $\lambda > 0$, a convex \fat{$\psi$} object $o$ can be rotated in each $ij$-plane by any angle $\alpha \leq \arccos{\paren*{1 - \frac{\lambda^2}{2{\fatD{o}}^2}}}$ while continuing to be entirely contained in its original $\lambda$-border.
\end{lemma}
\begin{proof}
    
    For some $1 \leq i < j \leq d$, consider the projection $\proj{ij}{o}$ of $o$ in the $ij$-plane, and let $C$ be the center of its \insc{} circle, as shows \cref{fig:fat-object-rotation}.
    Let $\alpha_{ij}$ be the angle by which we can rotate $o$ in the $ij$-plane until $\proj{ij}{o}$ touches its original $\border{\lambda}{ij}$ at some point $Q'$ (in case it does not touch, then $\alpha_{ij}$ is unconstrained), and let $Q$ be the equivalent point of $Q'$ in the original object $o$.
    We call $x$ the distance between $Q$ and $Q'$ and $y$ the distance between $C$ and $Q$ (and $Q'$).
    Note that $y \leq \fatD{o}$ and due to the choice of $Q$ and $Q'$ it holds that $x \geq \lambda$. Thus, by law of cosines on triangle $CQQ'$, we have that
    \begin{equation*}
        \cos\paren*{\alpha_{ij}} 
            = 1 - \frac{x^2}{2y^2}
            \leq 1 - \frac{\lambda^2}{2\fatD{o}^2},
    \end{equation*}
    which implies that $\alpha_{ij}$ is at least $\arccos{\paren*{1 - \frac{\lambda^2}{2{\fatD{o}}^2}}}$.
    Therefore, $o$ can be rotated in each $ij$-plane by any angle at most $\arccos{\paren*{1 - \frac{\lambda^2}{2{\fatD{o}}^2}}}$ while continuing to be within its original $\border{\lambda}{}$.
    \begin{figure}[htb!]
        \centering
        \begin{tikzpicture}[every node/.style={transform shape}, scale=0.8]
\tikzset{
pics/object/.style={
    code = {
        \begin{scope}

            \coordinate (-left)  at (-2.4,0);
            \coordinate (-right) at (4,0);
            \coordinate (-up)    at (0,2);
            \coordinate (-down)  at (0,-2);
            
            \path[looseness=0.7]
            (-left)  edge[bend left=45]    (-up)
            (-up)    edge[out=  0, in= 90] (-right)
            (-right) edge[out=270, in=  0] (-down)
            (-down)  edge[bend left=45]    (-left);
        \end{scope}
    }
},
pics/aug-object/.style={
    code = {
        \begin{scope}

            \coordinate (-left)  at ($ (-2.4,0) + (-#1,0) $);
            \coordinate (-right) at ($ (4,0) + (#1,0) $);
            \coordinate (-up)    at ($ (0,2) + (0,#1) $);
            \coordinate (-down)  at ($ (0,-2) + (0,-#1) $);
            
            \path[looseness=0.7]
            (-left)  edge[bend left=45]    (-up)
            (-up)    edge[out=  0, in= 90] (-right)
            (-right) edge[out=270, in=  0] (-down)
            (-down)  edge[bend left=45]    (-left);
        \end{scope}
    }
}
}
\begin{scope}


    \coordinate (C) at (0,0);

    \fill (C) circle (2pt);
    \node[left=0.1 of C] {$C$};
    \draw[dotted] (C) circle (1.95);


    \def \aug {20pt}

    \pic (obj) at (C) {object};
    \pic[dashed, rotate=-15] (obj-rotated) at (C) {object};
    \pic at (C) {aug-object={\aug}};


    \coordinate (Q1) at (-11:3.75);
    \coordinate (Q2) at (-26:3.75);

    \fill (Q1) circle (2pt);
    \node at ($ (Q1)+(-0.2, 0.3) $) {$Q$};
    \fill (Q2) circle (2pt);
    \node at ($ (Q2)+(0.3,-0.2) $) {$Q'$};

    \path  (C) edge node[pos=0.5, above] {$y$} (Q1) 
               edge node[pos=0.5, below] {$y$} (Q2)
          (Q1) edge node[pos=0.15, below=0.15] {$x$} (Q2);

    \pic["$\alpha_{ij}$", draw=black, angle eccentricity=1.35, angle radius=30] {angle = Q2--C--Q1};

    
    \draw[dotted] (Q1) circle (16pt);
    \draw[dotted] (Q1) edge node[pos=0.15, right=0.1] {$\lambda$} +(-50:16pt);

\end{scope}
\end{tikzpicture}
        \caption{
            Schematic of object $o$ projected in plane $ij$.
            The smaller solid object represents $\proj{ij}{o}$, the larger solid one is its $\border{\lambda}{ij}$, the dashed object is $\proj{ij}{o}$ rotated by $\alpha_{ij}$, and the dotted circle is the \insc{} sphere of $\proj{ij}{o}$, centered at $C$.
        }
        \label{fig:fat-object-rotation}
    \end{figure}
\end{proof}

Combining the previous lemmas, we prove that, for any $\eps > 0$, there is a packing in a bin augmented by a factor of $\eps$ and that respects some angle $\alpha$ dependent on $\eps$.

\begin{theorem}
\label{thm:packing-with-rotation}
    Let $B$ be a $d$-dimensional bin of size $(\range{l_1}{l_d})$.
    Consider a value $\gamma > 0$ and a set $\cali$ of convex $d$-dimensional \fat{$\psi$} objects, with $\fatd{o} \geq \gamma l_1$ for every $o \in \cali$, that fits in $B$. 
    For any $\eps > 0$, there exists a packing of $\cali$
    in an augmented bin of size $(\range{l'_1}{l'_d})$ with $l'_i \leq (1 + \eps) l_i$ for $i \in [d]$, and such that each object of $\cali$ has a rotation that respects $\alpha(\eps, d, \psi, \gamma) := \arccos{\paren*{ 1 - \frac{\eps^2}{8} ( d\psi + (1 + \sqrt{d})\psi \gamma )^{-2} }}$.
\end{theorem}
\begin{proof}
    Let $\delta = \eps l_1 / \paren*{\frac{d}{\gamma} + \sqrt{d} + 1}$.
    From \cref{thm:sparsed-packing}, we have that there is a packing $P'$ of $\cali$ with pairwise distance of at least $\delta$ in a bin $B'$ of size $(\range{l'_1}{l'_d})$ such that $l'_i \leq \paren*{ 1 + \frac{d \delta}{\gamma l_1} } l_i + \delta(1 + \sqrt{d}) \leq (1 + \eps) l_i$.
    Due to this spacing guarantee, the $\border{(\delta/2)}{}$ of the objects of $\cali$ do not overlap in $P'$.
    Then, from \cref{thm:fat-object-rotation}, we know that each object of $\cali$ can be rotated by any angle of at most $\alpha = \arccos\paren*{ 1 - \frac{\delta^2}{8D_o^2} }$ which can be shown to be at most $\arccos{\paren*{ 1 - \frac{\eps^2}{8} ( d\psi + (1 + \sqrt{d})\psi \gamma )^{-2} }}$ by using the facts that $D_o \geq \psi \gamma l_1$ and that the $\arccos\paren*{1-x}$ function is monotonically increasing on $x$.
    Thus, the objects in the packing $P'$ can be rotated so as to respect such angle while maintaining the feasibility of the packing.
\end{proof}

Finally, we show that our framework can produce a packing as described in \cref{thm:packing-with-rotation}, i.e., a packing in an augmented bin where each object has a rotation that respects $\alpha$.

The gap-structured partition and the classification of the objects into medium items and \nonmedium{} items remain the same.
To obtain a packing of the \nonmedium{} items, the idea is to account for the possible rotations an object can assume along with the configurations. 
For each configuration, we first fix one rotation for each object, then we apply the algorithm to check its feasibility. 
If a configuration is feasible under some combination of rotations of its objects, it suffices to find one. 
We must guarantee that the number of configurations is still bounded by a constant.
For each level $j\geq 0$, the number of configurations without rotations is bounded by a constant. 
Then, it suffices to have a constant bound on the number of rotations an object can assume.
To that purpose, we discretize the range $[0,2\pi]$ into a set $A = \{0, \alpha, 2\alpha, 3\alpha, \ldots, K\alpha\}$, with $K = \floor{2\pi/\alpha}$, of multiples of $\alpha$. 
Since a rotation is composed of $\binom{d}{2}$ values, each assuming a value from $A$, the number of possible rotations that respects $\alpha$ is bounded by $|A|^{\binom{d}{2}}$, which is constant under the assumption that $\eps$, $\Psi$, $d$ and $\gamma$ are constants.

To obtain a packing of the medium items, the process is the same as in \cref{sec:ras-for-fat-objects}, i.e., encapsulate the objects in hypercubes and use the $d$-dimensional version of \nfdh{}.
The difference here is that, for each object $o$, instead of considering the hypercube enclosing the object's \circums{} sphere, we consider the hypercube enclosing a sphere of diameter $2\gamma\fatD{o}$, i.e., the \circums{} sphere of the object plus its $\gamma$-border. This way, each object can be rotated to assume a rotation that respects $\alpha$ without surpassing the boundaries of its respective hypercube.

\subsection{Adding Item Multiplicity}
\label{sec:problems_with_multiplicity}

In general, adding item multiplicity brings more difficulty to the problems. Two notable examples are the cutting stock and the pallet loading problems.
In our framework, adding the additional feature of item multiplicity is straightforward. In fact, it requires no changes at all to the model. 
This is a great advantage over other techniques, such as dynamic programming, where it may not always be clear if it is possible to handle item multiplicity, and even if it is, most probably some more ad-hoc adaptations may be necessary.
In this sense, our framework is neat and very intuitive.
For instance, without any changes, it achieves the same results for all the problems described previously in their versions with item multiplicity. 

Most problems with item multiplicity are in \textrm{EXPSPACE}.
Therefore, to give a solution, we use a concise and complete representation of a packing: it suffices to know which configurations are used and their multiplicity. 
In a work regarding the cutting stock problem, Cintra et al.~\cite{CintraEtal_2007_approximability-cutting-stock} presented the ideas behind using short descriptions to represent a packing. 
In the following, we summarize the ideas adapted to our context and we refer the reader to their work for more details.
Formally, a \emph{description} of a solution is a list $\cald$ of pairs $(B, b_B)$, where $B$ is a bin type, which in our context are determined by the configurations, and $b_B$ is the number of bins of type $B$ that are used in the solution. A description $\cald$ is said to be a \emph{short description} if the bin types are all distinct and the size of $\cald$ (in terms of its representation in a given numerical system) is polynomially bounded on the size of the instance of the problem. 

Recall that, by the nature of our algorithm, two equal configurations of level~$0$ may result in different bin types depending on which configurations of levels~$j\geq1$ are packed in their empty volumes. 
To ensure a short description, we need to keep some control over the number of different bin types are created as we fill in the empty space of the configurations of level~$0$.
To this end, we first group equal configurations in each level, then we fill the empty volume in the configurations of level~$0$ in a way to keep control over the number of new bin types that are created. We do this until we have a description of the packing, i.e., the different bin types of level~$0$ and their multiplicity.
Let $(\calb_j^1, \ldots, \calb_j^{k_j})$ the list of different groups of configurations used in level $j$, for $j \geq 0$. 
At start, the description $\cald$ consists of the groups $\calb_0^{1},\ldots,\calb_0^{k_0}$, i.e., the bins only with items of level~$0$. Then, the subbins of level $j$ are packed in the free space of bins from level~${j-1}$, respecting the following rules: All the subbins (of level~$j$) of the same group are packed in sequence; bins (of level~${j-1}$) are opened by demand and once opened a bin of group $\calb_{j-1}^t$, for some $1\leq t\leq k_{j-1}$, all bins of this group are used before opening new bins of group $\calb_{j-1}^{t+1}$.
Suppose we are packing subbins of group $\calb_j^{t}$ in bins of group $\calb_{j-1}^q$, for some $j \geq 1$, $1 \leq t \leq k_j$ and $1 \leq q \leq k_{j-1}$. 
When the last subbin of $\calb_j^{t}$ is packed, there are three scenarios.
One, all bins of group $\calb_{j-1}^{q}$ are completely used; in this case, no new bin type is created. 
Two, only one bin of group $\calb_{j-1}^{q}$ is not completely used; in this case, one new bin type is created. 
Three, some bins of group $\calb_{j-1}^{q}$ are completely used, one is partially used and some were not opened; in this case, two new bin types are created. 
With simple calculations, it is possible to determine which of these scenarios happen.
Note that, for each one of the groups $\calb_j^1, \ldots, \calb_j^{k_j}$, $j \geq 1$, at most two new bin types are added in the description $\cald$. Since the number of groups is polynomial on the size of the instance, and moreover, we never create repeated types, the description $\cald$ created by this procedure is a short description of the solution given by the framework.
Finally, recall that the configurations refer to the rounded objects. To obtain a description with actual objects, it suffices to apply the same reasoning described for configurations, over the different sizes of objects.
With this, we have a short description of the \nonmedium{} objects.
It remains to give a short description of the medium objects. 

In the same work aforementioned, Cintra et al.~\cite{CintraEtal_2007_approximability-cutting-stock} argued that any algorithm for the bin packing problem that respects some properties, concerning the grouping of the objects by sizes and the order in which they are packed, yields an algorithm to give a short description of a solution to the cutting stock problem. They show that the \nfdh{} algorithm is one of these algorithm. 
Thus, we conclude that it is also possible to have a short description of the packing of the medium objects.

\subsection{Handling Additional Constraints in Knapsack Problems}
\label{sec:generalization-knapsack-problem}

Now we show that when concerning knapsack problems, our framework can actually deal with a more generalized version in which we can have additional constraints on the items.
In general, the allowed constraints are those that can be expressed as linear inequalities of the form $az \leq g$, with $a \geq 0$ and $g \geq 0$, where $z$ corresponds to the decision variables that decide whether each item is packed or not.
Constraints of this form can easily model common impositions in packing problems, such as:

\begin{itemize}
    \setlength\itemsep{0.5em}
    \item \textit{Conflict constraints}: $z_i + z_j \leq 1$ for each conflict between items $i$ and $j$;
    \item \textit{Multiple-choice constraints}: $\sum_{i \in F} z_i \leq 1$ for each class $F$ of items;
    \item \textit{Capacity constraints}: $\sum_{i \in [n]} w_{ji} z_i \leq W_j$ for each resource $j$.
\end{itemize}

Let $\calq = \srange{Q_1}{Q_q}$ be a set of $q$ constraints of this type.
We denote the $k$th constraint $Q_k$ by $\sum_{i \in [n]} a_{ki} z_i \leq g_k$.
Let $V = \set{z_i : i \in [n], \sum_{k \in [q]} a_{ki} > 0}$ be the set of the relevant variables present in $\calq$, with $v = \size{V}$.
We show next that if $q$ and $v$ are bounded by constants, then we can obtain an almost optimal packing respecting these constraints.
Note that these conditions are needed to handle such constraints.
To exemplify, if $v$ is unconstrained, then we can use the capacity constraints to model the vector multidimensional knapsack problem, which does not admit an \eptas{} even with only $2$-dimensional vectors unless $\classP = \classNP$.
Moreover, if $q$ is unconstrained, then the conflict constraints allow us to formulate the independent set problem, which does not even admit a $(1/n^{1 - \eps})$-approximation for any $\eps > 0$ unless $\classP = \classNP$.

First, by choosing a set of medium items of low profit, we can simply discard them without affecting the feasibility of any original solution.
Thus, we can restrict our attention to the \nonmedium{} items and obtain a super-optimal solution for them.
For that we can employ the same configuration IP $\Frounded$, fixing a configuration $C_0 \in \calc_0$:
%
\begin{alignat*}{4} 
 (\Frounded[C_0]) \quad & \omit\rlap{$\max \quad \displaystyle \smashoperator[lr]{ \sum_{s_i \in \struct{\cali}{t}} } z_i \profit{i}$} \nonumber \\
 & \mbox{s.t.} && \quad & \text{\eqref{eq:frounded_demands}}&\text{--\eqref{eq:frounded_x},}   & \qquad & \nonumber \\
 &             &&       &  x_0^{C_0} &= 1,                                                     &        & \\
 &             &&       &  \sum_{i \in [n]} a_{ki} z_i &\leq g_k                               &        & \forall\, k \in [q].
\end{alignat*}

We want to decompose the IP in blocks in the same manner as explained in \cref{sec:ckp-packing-level}.
For that, we first obtain an optimal fractional solution $\Tilde{Z} = (\Tilde{x}, \Tilde{b}, \Tilde{z})$ of the linear relaxation of $\Frounded$ but maintaining the integrality of the variables in $V$. Since $v \in \bigO(1)$, we can do this in polynomial time.
Now consider a constraint $Q_k$.
We have that $\sum_{i=1}^n a_{ki} \Tilde{z_i} = \Tilde{g}_k$ for some $\Tilde{g}_k \leq g_k$.
We partition the left-hand sum of the equation based on the items of each level, as follows.
For each level $j$, let $\Tilde{g}_{kj} = \sum_{i \in [n] : s_i \in S_j} a_{ki} \Tilde{z}_i$.
Then we replace the constraint $Q_k$ by the set of constraints 
\begin{equation} \label{eq:gkj}
    \sum_{i \in [n]: s_i \in S_j} a_{ki} z_i = \Tilde{g}_{kj} \quad \forall\, j \geq 0.
\end{equation}

With this replacement and using the fractional solution $\Tilde{Z}$, we can obtain an IP where all the constraints indexed by a level $j$ have only variables $x$, $z$ and $b$ of the corresponding level.
Thus we can decompose it in blocks where each block corresponds to a subproblem for each level.
Namely, the $j$th block is given by the formulation $\Flevel[j]$ described below.
%
\begin{alignat*}{4}
\paren[\big]{\Flevel[j]} \quad & \omit\rlap{$\max \quad \displaystyle \sum_{s_i \in S_j} z_i \profit{i}$} \\
 & \mbox{s.t.} && \quad &  \sum_{C \in \calc_j} x_j^C c_k &\leq n_j^k   & \qquad & \forall\, k \in [T_j], \\
 &             &&       &  \smashoperator[r]{\sum_{s_i \in S_j : \bar{r}_i = t_j^k}} z_i &= \sum_{C \in \calc_j} x_j^C c_k  &        & \forall\, k \in [T_j], \\
 &             &&       &  \sum_{C \in \calc_j} x_j^C &= \sum_{C \in \calc_j} \Tilde{x}_j^C,                             &        & \\
 &             &&       &  \smashoperator[r]{\sum_{C \in \calc_j}} f_j(C) x_{j}^C &= \Tilde{b}_{j+1},                    &        & \\
 &             &&       &  \smashoperator[r]{\sum_{i \in [n]: s_i \in S_j}} a_{ki} z_i &= \Tilde{g}_{kj}                 &        & \forall\, k \in [q], \\
&             &&       &  x_j^C &\in \Z_+                                            &        & \forall\, C \in \calc_j, \\
 &             &&       &  z_i &\in \binary                                           &        & \forall\, s_i \in S_j.
\end{alignat*}

In possession of this decomposition, we can obtain an optimal solution $(x_j^*, z_j^*)$ for each level by solving $\Flevel[j]$ individually, still maintaining the integrality of the associated variables in $V$, and then merge the solutions to obtain an optimal fractional solution $(x^*, z^*, b^*)$ for the entire problem.
Such solution is feasible, since for each constraint $Q_k$ we have that
\begin{equation*}
    \sum_{i \in [n]} a_{ki} z^*_i = \sum_{j \geq 0} \sum_{i \in [n]: s_i \in S_j} a_{ki} z^*_i = \sum_{j \geq 0} \Tilde{g}_{kj} = \Tilde{g}_k \leq g_k.
\end{equation*}

Furthermore, $\Flevel[j]$ has $\bigO(T_j + q)$ constraints. By making $\eps \leq 1/q$, we have that $T_j + q \leq \reveps^3 \ln(\reveps) + \reveps = \bigO(\reveps^3 \ln(\reveps))$.
Thus the number of fractional variables in each level is bounded by the same order as in \cref{thm:bound-different-radii}.
Also, since we maintained the integrality constraints of all the variables in $V$, we know that all the $z_i$ variables presented in $\calq$ have integer values, and therefore any rounding does not affect these values, and consequently does not interfere in the feasibility of the constraints $\calq$.
Hence, we can apply the same algorithm in \cref{sec:ckp-packing-level} to obtain a super-optimal solution for the \nonmedium{} items. This leads to the following.

\begin{theorem}
    Let $\cali$ be an instance of the geometric multiple knapsack problem with added constraints $\calq$ of the form $az \leq g$ with $a \geq 0$ and $g \geq 0$, and let $V = \set{z_i : i \in [n], \sum_{k \in [q]} a_{ki} > 0}$. If $\size{\calq} \in \bigO(1)$ and $\size{V} \in \bigO(1)$, then if the items are convex fat objects there is an \eaugptas{}. If the items are hyperspheres, then there is a \ptas{}.
\end{theorem}

\section{Final Remarks}
\label{sec:final-remarks}

Geometric packing problems have been investigated for centuries in mathematics. A great example is the Kepler's conjecture for the packing density of $3$-dimensional spheres in the Euclidean space. 
In the field of approximation algorithms, however, we do not find so many works on sphere packing. Most results are for squares and rectangles, and their $d$-dimensional counterparts.
To help filling this gap, in this work we presented a framework that provides approximation schemes, namely \ptas{}, \eptas{} and \eras{}, for several geometric packing problems.
In \cref{tbl:results-comparison}, we present a comparison between the results we explicitly derived in this work to the previous best results in the literature. 
Unless otherwise stated, all the problems consider $d$-dimensional convex fat objects as items.

\begin{table}[h!]
\begin{tabular}{lll}
 Problem                      & Previously                                             & Our results                                                 \\ \hline
Hsphere-\textrm{MKP}          & \augptas{}\cite{ChagasEtal-2023_approx-scheme-hypersphere-knapsack} & \ptas{}                                                     \\ 
Fat objects-\textrm{MKP}      & -                                                     & \eras{},                                                     \\ 
\textrm{BPP}                  & \ptas{}, \ras{} (Hspheres)\cite{MiyazawaEtal2015a}     & \eras{} (fat objects)                                       \\ 
\textrm{MSPP}                 & -                                                     & \eptas{} (lifting property)                                                   \\ 
\textrm{MMCP}                 & -                                                     & \eptas{}                                                    \\ 
\textrm{Multi, Rot, Add}      & -                                                     & \eras{} and \eptas{}                                        \\
\end{tabular}
\caption{A comparison of our results to the best found in the literature previously to our work. All the problems consider $d$-dimensional convex fat objects as items, unless otherwise observed. The term "Hsphere" is short for \emph{hyperspheres} and "\textrm{Multi, Rot, Add}" stands for Multiplicity, Rotation and Additional constraints, respectively. The sign "-" indicates no results found.}
\label{tbl:results-comparison}
\end{table}

A great versatility of our framework is regarding the characterization of objects that are accepted in the instance. In fact, it suffices that the input objects are convex and fat. 
This causes a clear contrast to previous works, where the items are mainly restricted to polygons, some types of polytopes and hyperspheres.
Upon that, our framework also supports the introduction of rotation on the items, another good contrast with previous works, where only translation was allowed. Moreover, the framework supports $d$-dimensional settings.

An appealing feature of our framework, worth emphasizing, is its natural support for item multiplicity, a characteristic that often introduces significant complexity. For example, the geometric cutting stock problem is only known to belong to $\textrm{EXPSPACE}$.
Another notable advantage of our framework is its flexibility in handling additional constraints in knapsack problems, such as the widely used \textit{conflict}, \textit{multiple-choice}, and \textit{multiple capacity} constraints.
These capabilities confer a robustness to our approach that may be lacking in other methods. For instance, algorithms based on dynamic programming would often require ad hoc or problem-specific modifications to address such features, if possible at all. In contrast, our framework incorporates them naturally, with minimal to no changes to the underlying algorithm.

We believe that the framework we presented has great potential to be applied to other settings, for instance to scheduling problems and assignments problems.
Thus, we believe that the framework we derived is of independent interest and consists an advance in the field of approximation algorithms for geometric packing problems, as well for other classes of problems with a packing flavor.


\printbibliography{}

\newpage
\appendix

\section{Omitted Proofs}

\thmNfdhSpecificSmallItems*
\begin{proof}
    From \cref{thm:nfdh-height}, \nfdh{} is able to pack $\cali$ using height
    \begin{align*}
        h' &= \frac{ \area{\cali} - \bar{s}^2 }{ w - \bar{s} } + \bar{s} \\
            &\leq \frac{ \area{\cali} }{ w - \bar{s} } + \bar{s} \\
            &\leq \frac{ \alpha \eps wh }{ w - \beta \eps^2 w} + \beta \eps^2 w \\
            &= \frac{\alpha}{1 - \beta \eps^2} \eps h + \beta \eps^2 w.
    \end{align*}
    Since $\eps \leq 1/4$ and $w \leq h$, we have that $\displaystyle{} \frac{\alpha}{1 - \beta \eps^2} \leq \frac{16 \alpha}{16 - \beta}$ and $\displaystyle{} \beta \eps^2 w \leq \frac{\beta}{4} \eps h$. By replacing these factors, we obtain that
    \begin{align*}
        h' &\leq \frac{16 \alpha}{16 - \beta} \eps h + \frac{\beta}{4} \eps h \\
            &= \frac{64\alpha + 16\beta - \beta^2}{64 - 4\beta} \eps h.
    \end{align*}
\end{proof}

\thmMediumItemsBinHeight*
\begin{proof}

    From the facts that $\area{H_t^*} \leq 2\eps wh$ and $t \geq 1$, we have that $\squarehull{{H_t^*}}$ is composed of squares of side length at most $\eps^\reveps w$ and $\area{\squarehull{{H_t^*}}} \leq \frac{8}{\pi} \eps w h$.
    From \cref{thm:nfdh-height}, the \nfdh{} algorithm would pack $\squarehull{{H_t^*}}$ in a bin of width $w$ and height $h'$ such that
    \begin{align*}
        h' &= \frac{ \area{\squarehull{{H_t^*}}} }{ w - \eps^\reveps w } + \eps^\reveps w \\
        &\leq \frac{8}{\pi} \cdot \frac{\eps w h}{w - \eps^\reveps w} + \eps^\reveps w \\
        &= \frac{8}{\pi} \cdot h \paren*{ \frac{\eps}{1 - \eps^\reveps} } + \eps^\reveps w \\
        &\leq \paren*{ \frac{8}{\pi} \cdot \frac{\eps}{1 - \eps^\reveps} + \eps^\reveps } h
    \end{align*}
    from the fact that $w \leq h$. Now, since $\eps \leq 1/4$, we have that $\eps^\reveps \leq \frac{1}{64} \eps$ and $\frac{1}{1 - \eps^\reveps} \leq \frac{256}{255}$. Making these replacements, we obtain that
    \begin{equation*}
        h' \leq \paren*{ \frac{8}{\pi} \cdot \frac{256}{255} \eps + \frac{1}{64} \eps } h \leq 3\eps h.
    \end{equation*}

    Therefore, $\squarehull{{H_t^*}}$, and consequently ${H_t^*}$, fits in a bin of size $w \times 3 \eps h$.
\end{proof}

\thmPackingMediumItems*
\begin{proof}
    Consider $H_t = (x_1, x_2, \dots)$ ordered in non-increasing order of $\profit{i}/d_i$.
    Since we know from \cref{thm:medium-items-bin-height} that ${H_t^*}$ fits in a bin of size $w \times 3 \eps h$, we define $B'$ as a slightly bigger bin of size $(1 + \eps) w \times 4 \eps h$.
    We define $H_t^k = (\range{x_1}{x_k})$, that is, the first $k$  items of $H_t$, $\nfdh{}(H_t^k)$ the packing obtained by \nfdh{} from trying to pack $\squarehull{H_t^k{}}$ into $B'$, and $D(H_t^k)$ the empty area in $\nfdh{}(H_t^k)$.
    
    Let $j+1$ be the smallest index in which \nfdh{} is not able to pack all items of $H_t^{j+1}$ into $B'$.
    If $x_{j+1} \notin \nfdh{}(H_t^{j+1})$, then $\nfdh{}(H_t^{j+1}) = \nfdh{}(H_t^{j})$.
    Otherwise, if $x_{j+1} \in \nfdh{}(H_t^{j+1})$, we know that the occupied area of 
    $\nfdh{}(H_t^{j})$ must be greater than the area of $\nfdh{}(H_t^{j+1})$ minus the area of $\squarehull{x_{j+1}}$.
    Thus in both cases we have that
    \begin{align*}
        \area{\nfdh{}(H_t^{j})} &> \area{\nfdh{}(H_t^{j+1})} - \area{\squarehull{x_{j+1}}} \\
        &= \area{B'} - D(H_t^{j+1}) - \area{\squarehull{x_{j+1}}} \\
        &\geq \area{B'} - \eps^\reveps w \cdot 2 [(1 + \eps) w + 4\eps h]/2 - \area{\squarehull{x_{j+1}}} & (\cref{thm:nfdh-empty-volume}) \\
        &\geq (1 + \eps) w \cdot 4 \eps h - \eps^\reveps w \cdot [(1 + \eps) w + 4\eps h] - (\eps^\reveps w)^2 \\
        &\geq 4 \eps (1 + \eps) wh - \eps^\reveps w \cdot (1 + 5\eps) h - \eps^{2\reveps} wh & (w \leq h) \\
        &= \brackets*{ 4(1 + \eps) - \eps^{\reveps - 1}(1 + 5\eps) - \eps^{2\reveps - 1} } \eps wh \\
        &\geq \brackets*{ 4(1 + \eps) - \frac{1}{4^3}(1 + 5\eps) - \frac{1}{4^7} } \eps wh & (\eps \leq 1/4) \\
        &\geq \frac{8}{\pi} \eps wh \\
        &\geq \area{\squarehull{{H_t^*}}}.
    \end{align*}
    
    This means that $\nfdh{}(H_t^{j})$ fills an area at least as big as that of $\squarehull{{H_t^*}}$, and due to the ordering in $H_t$ such area is filled with the items of highest relative value. Thus, $\profits{\nfdh{}(H_t^{j})} \geq \profits{{H_t^*}}$.
    
    Lastly, we move all circles that are entirely contained in the rightmost strip of length $2\eps w$ to a new bin of size $w \times 4\eps h$,
    which leaves the rightmost area of $\eps w \times 4 \eps h$ empty since all circles have diameter at most $\eps^\reveps w$.
    Then, by stacking this new bin on top of $B'$ we obtain a packing of the circles in a bin of size $w \times 8\eps h$, as shows \cref{fig:packing-medium-items}.
\end{proof}

\begin{figure}[htb!]
    \centering
    \scalebox{0.4}{
\begin{tikzpicture}

\definecolor{customblue}   {RGB}{161, 197, 255}
\definecolor{customred}    {RGB}{255, 157, 144}
\definecolor{customorange1}{RGB}{255, 230, 153}
\definecolor{customorange2}{RGB}{255, 240, 196}
\definecolor{customgreen1} {RGB}{186, 255, 185}
\definecolor{customgreen2} {RGB}{211, 255, 209}
\definecolor{custompink1}  {RGB}{255, 191, 225}
\definecolor{custompink2}  {RGB}{252, 215, 235}

\def \rightslicex {7.4}

\tikzset{
pics/item/.style n args={3}{
    code = {
        \begin{scope}
            \draw[fill=#2] (0,0) rectangle (#1, #1);
            \draw[fill=#3] (#1/2, #1/2) circle (#1/2);
            \node (-left)  at ( 0,0) {};
            \node (-right) at (#1,0) {};
        \end{scope}
    }
},
pics/listitem/.style n args={3}{
    code = {

        \readlist \items {#1}
        \edef \n {\listlen\items[]}
        
        \pic (c1) at (0,0) { item={\items[1]}{#2}{#3} };
        \ifthenelse{\n > 1}{
            \foreach \i in {2, ..., \n} {
                \edef \length {\items[\i]}
                \pgfmathsetmacro{\j}{int(\i - 1)}
                \pic (c\i) at ($(c\j-right)$) { item={\length}{#2}{#3} };
            }
        }{}
        \node (-end) at (c\n-right) {};
    }
},
pics/packingw/.style={
    code = {

        \draw[fill=customblue] (0,0) rectangle (8.4,5);

        \pic (-level1) at (0,0)   { listitem={2, 2, 1.8, 1.8}{customorange1}{customorange2} };
        \pic (-level2) at (0,2)   { listitem={1.5, 1.4, 1.3, 1.3, 1.2, 1.2}{customorange1}{customorange2} };
        \pic (-level3) at (0,3.5) { listitem={1.2, 1.1, 1, 0.9, 0.8, 0.8, 0.7, 0.7, 0.65}{customorange1}{customorange2} };

        \node (-tl) at (0,5) {};
        \node (-tr) at (8.4,5) {};
        \node (-bl) at (0,0) {};
        \node (-br) at (8.4,0) {};
    }
}
}

    \pic (bin1) at (0,0) {packingw};

    \node (augmentation-bl) at (8.4, 0) {};
    \draw[fill=customred] (augmentation-bl) rectangle ++ (1,5);

    \pic[local bounding box=pinkitems1] at (bin1-level1-end) { listitem={1.6}{custompink1}{custompink2} };
    \pic (pinkitems2) at (bin1-level2-end) { listitem={1.2}{custompink1}{custompink2} };
    \pic (pinkitems3) at (bin1-level3-end) { listitem={0.6, 0.6}{custompink1}{custompink2} };

    \draw[red, dashed, line width=1mm] (\rightslicex, 0) -- (\rightslicex, 5.65);

    \draw[Stealth-Stealth] ($ (bin1-bl) + (-0.5, 0) $) -- node[left]  {$4\varepsilon h$} ($ (bin1-tl) + (-0.5, 0) $);
    \draw[Stealth-Stealth] ($ (bin1-bl) + (0, -0.5) $) -- node[below] {$w$} ($ (bin1-br) + (0, -0.5) $);
    \draw[Stealth-Stealth] ($ (augmentation-bl) + (0, -0.5) $) -- node[below] {$\varepsilon w$} ++ (1, 0);
    \draw[Stealth-Stealth] (\rightslicex, 5.5) -- node[above] {$2 \varepsilon w$} ++ (2, 0);

    \pic (bin2) at ($ (bin1-br) + (2,0) $) {packingw};

    \draw[fill=customred] (bin2-tl) rectangle ++ (8.4, 5);

    \coordinate (bin3-shift) at (3, 5);
    \pic at ($ (bin1-level1-end) + (bin3-shift) $) { listitem={1.6}{custompink1}{custompink2} };
    \pic at ($ (bin1-level2-end) + (bin3-shift) $) { listitem={1.2}{custompink1}{custompink2} };
    \pic at ($ (bin1-level3-end) + (bin3-shift) $) { listitem={0.6, 0.6}{custompink1}{custompink2} };

    \draw[Stealth-Stealth] ($ (bin2-bl) + (0, -0.5) $)   -- node[below] {$w$} ($ (bin2-br) + (0, -0.5) $);
    \draw[Stealth-Stealth] ($ (bin2-br) + ( 0.5, 0) $) -- node[right]  {$8\varepsilon h$} ($ (bin2-tr) + (0.5, 5) $);

\end{tikzpicture}
}
    \caption{Transforming the packing of medium items in a bin of size $(1 + \eps) w \times 4\eps h$ obtained by \nfdh{} into a packing in a bin of size $w \times (1 + 8\eps h)$. }
    \label{fig:packing-medium-items}
\end{figure}

\thmNumberConfigurationsPolynomial*
\begin{proof}
    Since the circles of $S_j$ have a minimum radius $\rmin{j} \geq \eps^{\reveps^2 - \reveps + 1} w_j / 2$, the maximum number of circles that fit in a bin of level $j$ is bounded by
    \begin{equation*}
    \frac{w_j h_j}{\area{\rmin{j}}}
        = \frac{w_j h_j}{\pi(\rmin{j})^2}
        \leq \frac{w_j h_j}{\pi (\eps^{\reveps^2 - \reveps + 1} w_j/2)^2}
        = \frac{4}{\pi} r^{2\reveps^2 - 2\reveps + 2} \frac{h_j}{w_j}
        := \constantsizeconfig{j},
    \end{equation*}
    which is constant under the assumption that $h/w \in \bigO(1)$. Then, the maximum number of feasible configurations of $S_j$ is at most $\binom{n}{\constantsizeconfig{j}} \in \bigO(n^{\constantsizeconfig{j}})$, thus polynomial in $n$.
\end{proof}

\thmBoundDifferentRadii*
\begin{proof}
    Using the fact that for any number $x > -1$ it holds that
    \begin{equation*}
        \frac{x}{1+x} \leq \ln(1 + x) \leq x,
    \end{equation*}
    we have 
    \begin{equation*}
        \log_{1 + \eps} \reveps
        = \frac{\ln(\reveps)}{\ln(1 + \eps)}
        \leq \paren*{1 + \frac{1}{\eps}} \ln(\reveps),
    \end{equation*}
    which implies 
    \begin{equation}\label{eq:bound-different-radii_log<=ln}
    \log_{1+\eps} \reveps \leq (\reveps+1)\ln(\reveps).
    \end{equation}

    Recall that the radii of the circles of $S_j$ are rounded up to values of the set $ R_j = \set{ \rmin{j}(1 + \eps)^k : k \geq 0, \rmin{j}(1 + \eps)^k < \rmax{j} } \cup \set{\rmax{j}} $.
    Since $T_j$ is bounded by $\size{R_j}$ and $k \leq \ceil{\log_{1 + \eps} (\rmax{j} / \rmin{j})}$, we have that
    \begin{align*}
        T_j &\leq \log_{1 + \eps} \paren*{ \frac{\rmax{j}}{\rmin{j}} } + 2 \\
        &= \log_{1 + \eps} (\reveps^{\reveps(\reveps-1)}) + 2 \\
        &= \reveps (\reveps-1) \log_{1 + \eps} (\reveps) + 2 \\
        &\leq \reveps (\reveps-1) (\reveps+1) \ln(r) + 2  \qquad \text{(from inequality~\eqref{eq:bound-different-radii_log<=ln})} \\
        &= \reveps^3 \ln(r) - \reveps \ln(r) + 2 \\
        &\leq \reveps^3 \ln(r)
    \end{align*}
    for $\eps \leq 1/4$.
\end{proof}

\thmBoundNumberConfigurations*
\begin{proof}
    The bound $\constantsizeconfig{j}$, defined in \cref{thm:number-configurations-polynomial}, on the maximum number of circles that fit in one bin of level~$j$ still holds after rounding their radii.
    Then, since a configuration is composed of $T_j$ values, and each value can range from $0$ to $\constantsizeconfig{j}$, the total number of possible configurations is at most
    \begin{equation*}
        (\constantsizeconfig{j}+1)^{T_j} \leq \paren*{ \frac{4}{\pi} \reveps^{2\reveps^2-2\reveps+2} \frac{h_j}{w_j} + 1}^{ r^3\ln(r)},
    \end{equation*}
    which is constant under the assumption that $h/w \in \bigO(1)$.
\end{proof}

\thmCheckConfigurationFeasibility*
\begin{proof}
    Let $\cals_C$ be the set of circles of configuration $C$.
    Since the number of different radii is constant, from~\cref{thm:bound-different-radii}, and the radii of the circles of $\cals_C$ are at least $r_{\min}^j$,
    we can use the algorithm from \cref{thm:fkm-etal_bin-packing-in-augmented-bins} with $\gamma$ to obtain a solution for the \cbp{} instance $(\cals_C,w_j,h_j)$.
    If such solution uses more than one bin, we say the configuration $C$ is unfeasible.
    Otherwise, the solution consists of a packing of the circles of $C$ into exactly one bin of size $w_j \times (1+\gamma)h_j$. 
    The algorithm from \cref{thm:fkm-etal_bin-packing-in-augmented-bins} runs in polynomial time on the number of circles. Since the number of circles of a configuration of scaled circles is bounded by a constant, from \cref{thm:bound-number-configurations-scaled}, the algorithm takes constant time.
\end{proof}

\thmFroundedFexact*
\begin{proof}
    Given a level~$j$, for each configuration $C \in \calc_j$ of the scaled circles of $S_j$, let $\calr_j^C$ be the set of configurations of the original circles of $S_j$ that, after scaling the circles, became equivalent to $C$.
    Let $(\Hat{x}, \Hat{b}, \Hat{z})$ be an optimal solution of $\Fexact$ for $(\struct{\cali}{t}, w, \Hat{h}, p)$. 
    We build a solution $(x, b, z)$ of $\Frounded$ as follows: $b = \Hat{b}$, $z = \Hat{z}$, and $x_j^C = \sum_{D \in \calr_j^C} \Hat{x}_j^D$ for each $j\geq 0$ and $C \in \calc_j$.
    By construction, both solutions have the same objective value. It remains to prove the feasibility of $(x,b,z)$.
    
    Observe that by the definition of $x_j^C$ we have that
    \begin{equation*}
        \sum_{C \in \calc_j} x_j^C = \sum_{C \in \calc_j} \sum_{D \in \calr_j^C} \Hat{x}_j^D = \sum_{\Hat{C} \in \Hat{\calc}_j} \Hat{x}_j^{\Hat{C}}.
    \end{equation*}
    
    Thus, constraints \eqref{eq:frounded_demands}--\eqref{eq:frounded_bins} are satisfied by $(x,b,z)$.
    It only remains to show the satisfiability of constraints~\eqref{eq:frounded_free-area}.
    For that, fixing a level~$j$, consider a configuration $\Hat{C} \in \Hat{\calc}_j$ of original circles and its counterpart $C \in \calc_j$ of scaled circles.
    We show that $f_{j}(C) \geq \Hat{f}_{j}(\Hat{C})$.
    \begin{align}
       f_{j}(C) &= \frac{w'_{j-1} h'_{j-1} - (1+16\eps)\area{C}}{w'_j h'_j} \nonumber \\
       &\geq \frac{w'_{j-1} h'_{j-1} - (1+16\eps) (1+\eps)^2\area{\Hat{C}}}{w'_j h'_j} \nonumber \\
       &= \frac{(1+16\eps)(1+\eps)^2 w_{j-1} \Hat{h}_{j-1} - (1+16\eps) (1+\eps)^2\area{\Hat{C}}}{(1+16\eps)(1+\eps)^2 w_j \Hat{h}_j} \nonumber \\
       &= \frac{w_{j-1} \Hat{h}_{j-1} - \area{\Hat{C}}}{ w_j \Hat{h}_j} \label{eq:free-area-opt} \\
       &\geq \Hat{f}_{j}(\Hat{C}), \nonumber
    \end{align}
    since \cref{eq:free-area-opt} is an area-based upper bound on the number of subbins that fit in the empty space.
    With this result we obtain that for any level $j \geq 1$,
    \begin{align*}
        b_j &= \Hat{b}_j \\
            &\leq \sum_{ \Hat{C} \in \Hat{\calc}_{j-1} } \Hat{f}_{j-1}(\Hat{C}) \Hat{x}_{j-1}^{\Hat{C}} \\
            &= \sum_{C \in \calc_{j-1}} \sum_{ D \in \calr_{j-1}^C } \Hat{f}_{j-1}(D) \Hat{x}_{j-1}^D \\
            &\leq \sum_{C \in \calc_{j-1}} f_{j-1}(C) x_{j-1}^C.
    \end{align*}

    Thus, $(x,b,z)$ is a feasible solution for $\Frounded$ with same objective value as the solution $(\Hat{x}, \Hat{b}, \Hat{z})$ for $\Fexact$.
    Consequently, $\opt(\Frounded) \geq \opt(\Fexact)$.
\end{proof}

\thmBoundNumberRoundedUpVariables*
\begin{proof}
    Let $(x^*, b^*, z^*)$ be an optimal fractional solution of $\milpFrounded$ and let $(\Tilde{x}, b^*, \Tilde{z})$ be the solution obtained by \cref{alg:balanced-fractional-solution} for $(\struct{\cali}{t}, w, \Hat{h}, p)$.
    First note that the solution $(x_j^*, z_j^*)$ given by the restriction of $x^*$ and $z^*$ to the circles of $S_j$ is a feasible solution to the linear relaxation of $\Flevel[j][b^*]$.
    Thus the objective value of the obtained solution $(\Tilde{x}, b^*, \Tilde{z})$ is at least the value of the optimal one $(x^*, b^*, z^*)$. On the other hand, the value of $(\Tilde{x}, b^*, \Tilde{z})$ cannot be greater than the value of $(x^*, b^*, z^*)$. Otherwise, if for some $j \geq 1$ the value of $(\Tilde{x}_j, \Tilde{z}_j)$ is greater than the value of $(x_j^*, z_j^*)$, then we can simply replace the latter by the former and obtain a feasible solution to $\milpFrounded$ whose profit is greater than the one given by $(x^*, b^*, z^*)$, contradicting its optimality.
    Therefore, $(\Tilde{x}, b^*, \Tilde{z})$ is an optimal solution to $\milpFrounded$, and since $\Flevel[j]$ has $2T_j + 2$ constraints, $\Tilde{x}_j$ has at most $2T_j + 2$ non-null variables.

    Regarding the time complexity, the dominating step of \cref{alg:balanced-fractional-solution} is solving the linear relaxation of $\Frounded$.
    Given a MILP with $n$ integer variables, $d$ continuous variables and a binary encoding of size $L$, Lenstra's algorithm~\cite{Lenstra1983} solves the MILP with running time $2^{\bigO(n^3)} \poly(d, L)$.
    In $\milpFrounded$ there are $\bigO(n r^{r^6})$ variables in total, where $\bigO(r^{r^6})$ of them are integer.
    The largest coefficients present in $\Frounded$ are $c_k$ and $f_j(C)$, both bounded by $O(r^{r^2})$, and $n_j^k$, bounded by $O(n)$, therefore the binary encoding of $\Frounded$ is of size $L = O(\log(n) + r^2 \log(r))$.
    Thus, $\milpFrounded$ can be solved in time $\bigO\paren*{ 2^{r^{3 r^6}} \poly(n, r^r) }$.
\end{proof}

\thmStripForExtraBins*
\begin{proof}
    Since $(x^*, b^*, z^*)$ is balanced, from \cref{thm:bound-number-rounded-up-variables} we have that for each level $j \geq 1$, at most $2T_j + 2$ variables $x_j$ were rounded. 
    Recall from \cref{thm:bound-different-radii} that $T_j \leq \reveps^3 \ln(\reveps)$. 
    Thus, 
        $2T_j + 2 \leq 2\reveps^3\ln(\reveps) + 2$.
    
    Let $D$ be a rectangle of size $w' \times \eps h'$. Since we assume $w \leq h$, we have $\area{D} \geq \eps w'^2$.
    Recall that $w'_j = h'_j = \eps^{\reveps(t+(j-1)\reveps)+\reveps-1} w'$, and moreover, that $r \geq 4$ and $t \geq 1$.
    First we analyze level $1$ separately. 
    By definition, we have $w'_1 = \eps^{\reveps t + \reveps - 1} w' \leq \eps^{2\reveps - 1} w'$, because $t \geq 1$. 
    Then the number of bins of size $w'_1 \times h'_1$ that fit into $D$ is bounded by 
    \begin{equation*}
        \frac{\area{D}}{{w'}_1^2} \geq \eps \frac{{w'}^2}{{w'}_1^2} \geq \reveps^{4\reveps - 3} \geq 2T_1 + 3.
    \end{equation*}
    
    This means that $D$ is sufficiently large to accommodate all extra bins of level $1$, and it still has space for at least one more bin of size $w'_1 \times h'_1$.
    Similarly, we show that for $j \geq 2$, one bin of size $w'_{j-1} \times h'_{j-1}$ is sufficient to accommodate the extra bins of level~$j$ and it still has space for at least one more bin of size $w'_j \times h'_j$. 
    Note that $w'_j = \eps^{r^2} w'_{j-1}$. Then the result follows from direct calculation.
    \begin{equation*}
        \frac{{w'}_{j-1}^2}{{w'}_j^2} = r^{2r^2} \geq 2 T_j + 3.
    \end{equation*}
    
    Since after packing the extra bins of level~$1$ in $D$ it still has space for at least one free bin of level~$1$, and for level~$j \geq 2$, one bin of level~${j-1}$ is sufficient to pack all the extra bins plus one of level~$j$, we conclude that all the extra bins of every level fit into $D$.
\end{proof}

\thmConfigurationToPacking*
\begin{proof}
    For each configuration $C \in X_j$,  we use \cref{thm:check-configuration-feasibility} to obtain a packing of the scaled circles of $C$ in a bin of size $w'_j \times (1+\gamma)h'_j$, in constant time. 
    It remains to replace the scaled circles with original ones in such a way that the total profit is maximum. For that, it is enough to choose the original circles of highest profit, as follows.
    For each $k = 1,\ldots,T_j$, let $\eta_k$ be the total number of scaled circles of radius $t_j^k$ within the collection $X_j$. 
    We sort $S_j$ in non-increasing order of profit, and we substitute the $\eta_k$ scaled circles of radius $t_j^k$ with the $\eta_k$ original circles of highest profit among the ones whose rounded radius is $t_j^k$.
    If $\eta_k$ is greater than the number $n_j^k$ of original circles whose rounded radius is $t_j^k$, which may happen when $X_j$ comes from rounding up a fractional solution of $\milpFrounded$, it suffices to pack all such $n_j^k$ original circles; it trivially maximizes the profit coming from the circles of such radius since any optimal solution cannot use more than $n_j^k$ of them.
    This procedure can be done in $\bigO(n\log n)$ time, where $n$ is the number of original circles. 
\end{proof}

\thmBoundDifferentRadiiLevelZero*
\begin{proof}
    First recall that $\rmax{0} \leq w/2$ and $\rmin{0} \geq \eps^{\reveps t} w/2$, and thus $\rmax{0} / \rmin{0} \leq r^{\reveps t} \leq r^{\reveps(\reveps - 1)}$ since $t \leq \reveps - 1$.
    Since $T_0$ is bounded by the size of $R_0$, we have that
    \begin{align*}
        T_0 &\leq \log_{1+\delta}\paren*{\frac{\rmax{0}}{\rmin{0}}} + 2 \\
        &\leq \log_{1+\delta}(\reveps^{\reveps(\reveps-1)}) + 2 \\
        &\leq \reveps(\reveps-1) \paren*{1 + \frac{1}{\delta}} \ln{\reveps} + 2 \\
        &\leq \paren*{1 + \frac{96}{\pi^2} r^{4\reveps^2 - 4\reveps + 6} \frac{h^2}{w^2}} \reveps (\reveps - 1) \ln{\reveps} + 2 \\
        &\leq 12 \reveps^{4\reveps^2 - 4\reveps + 9} \frac{h^2}{w^2}
    \end{align*}
    for $\eps \leq 1/4$, which is constant under the assumption that $h/w \in \bigO(1)$.
\end{proof}

\thmShiftingLevelZero*
\begin{proof}
    Let $p_i = (x_i, y_i)$ be the center position of the circle $i$ in $P$.
    Since $P$ is a packing, we know that $r_i \leq x_i \leq w - r_i$ and $r_i \leq y_i \leq h - r_i$ for any circle $i$, and $\dist{p_i}{p_j} \geq r_i + r_j$ for any two circles $i$ and $j$.
    Now consider the circles positioned as in $P$ but with the radius scaled up to the closest value of $R_0$.
    The scaled radius $\Bar{r}_i$ of any circle~$i$ can increase only by a factor of at most $1 + \delta$, that is, $\Bar{r}_i / r_i \leq (1 + \delta)$, which implies that $r_i \geq \Bar{r}_i / (1 + \delta)$.
    Furthermore, since $r_i \leq w/2 \leq h/2$, we have that $\Bar{r}_i \leq (1 + \delta) h/2$.
    Using these inequalities, we can show that the distance between two scaled circles in $P$ becomes
    \begin{align*}
        \dist{p_i}{p_j} &\geq r_i + r_j \\
            &\geq \frac{1}{1+\delta} (\Bar{r}_i + \Bar{r}_j) \\
            &= \Bar{r}_i + \Bar{r}_j - \frac{\delta}{1+\delta} (\Bar{r}_i + \Bar{r}_j) \\
            &\geq \Bar{r}_i + \Bar{r}_j + \delta h,
    \end{align*}
    and using the same reasoning, we can show that $\Bar{r}_i - \delta h \leq x_i \leq w - \Bar{r}_i + \delta h$ and $\Bar{r}_i - \delta h \leq y_i \leq h - \Bar{r}_i + \delta h$ for any circle $i$.
    Therefore, this attribution is a $\delta h$-packing of the scaled circles in $B_{w \times h}$.
    Using the result of \cref{thm:shifting-algorithm}, this $\delta h$-packing can be converted into a packing in a bin of width $w$ and height $(1 + n \sqrt{6\delta}) h \leq (1 + \constantsizeconfig{0} \sqrt{6 \cdot \eps^2/(6{\constantsizeconfig{0}}^2)} = (1 + \eps) h$.
\end{proof}

\thmEmptyAreaBigCircles*
\begin{proof}
    Consider the bottom left corner of $B$.
    Since the items are circles of radius at least $\delta$, the region delimited by this corner and a circle $C$ of radius $\delta$ positioned at $(\delta, \delta)$ is surely empty.
    By symmetry, the closest point of $C$ from the origin is a point $p = (a, a)$.
    Then we have that $\delta^2 = 2 (\delta - a)^2$, which implies $a \geq (1 - 1/\sqrt{2}) \delta$.
    Thus, the empty square of side length $a$ from the origin to $p$ gives us the desired bound on the empty area. 
\end{proof}

\thmFatObjectsBoundNumberConfigurations*
\begin{proof}
    To bound the number of objects that fit in a bin, we use the volume of the hypercube inscribed in  a hypersphere of diameter $\fatDmin{j}/\psi$ (i.e., the smallest possible \insc{} sphere among the scaled objects of $S_j$), which has side length $\frac{\fatDmin{j}}{\psi\sqrt{d}}$. 

    We estimate level~$0$ first. The maximum number of scaled objects of $S_0$ that fit in the knapsack is bounded by
    \begin{align*}
        \constantsizeconfig{0} &\leq \frac{\vol{B}}{\paren*{\frac{\fatD{j}^{\min}}{\psi \sqrt{d}}}^d} \\
        &\leq \paren*{ \psi\sqrt{d} }^d \frac{({l_{\max}})^d}{(\eps^{\reveps t} \ l_{\min})^d} \\
        &\leq \paren*{ \psi \reveps^{\reveps(\reveps - 1 )} \sqrt{d} \,\frac{l_{\max}}{l_{\min}} }^d \quad \text{(since $t \leq \reveps - 1$).}
    \end{align*}

    Since a configuration of $S_0$ is composed of $\tau T_0$ values, and each value can range from $0$ to $\constantsizeconfig{0}$, the total number of possible configurations of level~$0$ is at most 
    \begin{equation*}
        (\constantsizeconfig{0} + 1)^{\tau T_0} \leq
        \paren*{
            \paren*{ \psi \reveps^{\reveps(\reveps - 1 )} \sqrt{d} \,\frac{l_{\max}}{l_{\min}} }^d
            + 1
        }^{\tau \reveps^3 \ln{\reveps}},
    \end{equation*}
    which is constant under the assumption that
    $l_{\max}/l_{\min}$, $d$ and $\psi$ are constants.

    For levels~$j \geq 1$, the maximum number of objects that fit in a subbin of their level is bounded by 
    \begin{align*}
        \constantsizeconfig{j} &\leq \frac{{w_{j}}^d}{\paren*{\frac{\fatDmin{j}}{\psi \sqrt{d}}}^d} \\
        &= \paren*{ \psi\sqrt{d} }^d \frac{(\eps^{\reveps(t + (j-1)\reveps) + \reveps - 1} l_{\min})^d}{(\eps^{r(t + j\reveps)} l_{\min})^d} \\
        &= \paren*{ \psi \reveps^{ \reveps(\reveps - 1) + 1 } \sqrt{d} }^d.
    \end{align*}
    
    Then the total number of possible configurations in each level $j \geq 1$ is at most
    \begin{equation*}
        (\constantsizeconfig{j}+1)^{\tau T_j} \leq
        \paren*{
            \paren*{ \psi \reveps^{ \reveps(\reveps - 1) + 1 } \sqrt{d} }^d
            + 1
        }^{ \tau \reveps^3 \ln{\reveps} },
    \end{equation*}
    which is constant under the assumption that $d$ and $\psi$ are constants.
\end{proof}

\thmFatObjectsStripForExtraBins*
\begin{proof}
    
    Analogously as in \cref{thm:bound-number-rounded-up-variables}, the number of rounded $x_j$ variables is at most $2 \tau T_j + 2$, for each $j\geq 1$.
    Thus, similar as in \cref{thm:strip-for-extra-bins}, we have the following bound on the number of extra bins in each level $j \geq 1$,
    \begin{equation*}
        2 \tau T_j + 2 \leq 2 \reveps^3\ln(\reveps) + 2 \leq 2 \reveps^4 \ln(\reveps) + 2,
    \end{equation*}
    since $\eps \leq 1/\tau$.
    
    The number of bins of level~$1$ that fit in a strip of size $(l_1, \ldots, l_{d-1}, \eps l_d)$ is bounded by
    \begin{align*}
        \frac{ \eps l_{d} \prod_{i=1}^{d-1} l_i }{w_1^d} 
        &\geq \frac{\eps {l_{\min}}^d}{\eps^{d(2\reveps-1)} {l_{\min}}^d} & \text{(since $t \geq 1$)} \\
        &= \reveps^{d(2\reveps-1)-1} & \\
        &\geq \reveps^{13} & \text{(since $\reveps \geq 4$ and $d\geq2$)} \\
        &\geq 2\reveps^4\ln(\reveps) + 3.
    \end{align*}
    
    This means that the strip is sufficiently large to accommodate all extra bins of level $1$, still leaving room for at least one more bin of level~$1$.
    Similarly,
    for $j \geq 2$, one bin of level $j-1$ is sufficient to accommodate the extra bins of level~$j$, still leaving room for at least one more bin of level $j$. 
    \begin{align*}
        \frac{w_{j-1}^d}{w_j^d} &= \frac{\eps^{d(\reveps(t + (j-2)\reveps) + \reveps - 1)} {l_{\min}}^d} {\eps^{d(\reveps(t + (j-1)\reveps) + \reveps - 1)} {l_{\min}}^d} & \\
        &= \reveps^{d\reveps^2} & \\
        &\geq r^{32} & \text{(since $\reveps \geq 4$ and $d\geq2$)} \\
        &\geq 2\reveps^4\ln(\reveps) + 3.
    \end{align*}
    
    Thus all the extra bins of every level fit into a strip of size $(l_1, \ldots, l_{d-1}, \eps l_d)$.
\end{proof}

\thmRatioVolumeFatObjectHypercube*
\begin{proof}
    Let $V$ be the volume of the $d$-dimensional unit hypersphere.
    Denote by $V_k$ the volume of a $d$-dimensional sphere of diameter $k$. 
    We start with the following fact.

    \begin{claim}
        The volume of a $d$-dimensional hypersphere of diameter $k$ is $(k/2)^d V$.
    \end{claim}
    
    Observe that the hypercube $\squarehull{o}$ is circumscribed by a hypersphere of diameter $ k = \sqrt{2}\fatD{o}$. 
    Thus,
    \begin{equation*}
        \frac{\squarehull{V_o}}{V_o}
        \leq \frac{V_{k}}{V_{\fatd{o}}}
        = \frac{(k/2)^d V}{(\fatd{o}/2)^d V}
        = \frac{(\sqrt{2}\fatD{o}/2)^d V}{(\fatd{o}/2)^d V}
        = \paren*{\sqrt{2} \psi}^d.
    \end{equation*}
\end{proof}

\end{document}